\documentclass[pdflatex,sn-mathphys-ay]{sn-jnl}

\usepackage{graphicx}%
\usepackage{multirow}%
\usepackage{amssymb}
\usepackage{amsmath}
\usepackage{amsthm}
\usepackage{amsfonts}
\usepackage{mathrsfs}
\usepackage{mathtools}
\usepackage[title]{appendix}%
\usepackage{xcolor}%
\usepackage{textcomp}%
\usepackage{manyfoot}%
\usepackage{booktabs}%
\usepackage{algorithm}%
\usepackage{algorithmicx}%
\usepackage{algpseudocode}%
\usepackage{listings}%

\usepackage{tikz}
\usetikzlibrary{arrows.meta,positioning,fit,calc,decorations.pathreplacing,shapes.multipart}
\usepackage{cleveref}

\theoremstyle{thmstyleone}%
\newtheorem{theorem}{Theorem}[section]
\newtheorem{proposition}[theorem]{Proposition}
\newtheorem{lemma}[theorem]{Lemma}

\theoremstyle{definition}
\newtheorem{definition}[theorem]{Definition}
\newtheorem{assumption}[theorem]{Assumption}
\theoremstyle{remark}
\newtheorem{remark}[theorem]{Remark}

\begin{document}

\title[Dedifferentiation stabilizes stem cell lineages: \\ From CTMC to diffusion theory and thresholds]{Dedifferentiation stabilizes stem cell lineages: \\ From CTMC to diffusion theory and thresholds}

\author[1]{\fnm{Jiguang} \sur{Yu}}\email{jyu678@bu.edu}
\equalcont{These authors contributed equally to this work as co-first authors.}

\author*[2]{\fnm{Louis Shuo} \sur{Wang}}\email{swang116@vols.utk.edu}
\equalcont{These authors contributed equally to this work as co-first authors.}

\author[3]{\fnm{Ye} \sur{Liang}}\email{zcahiad@ucl.ac.uk}
\equalcont{These authors contributed equally to this work as co-first authors.}

\affil[1]{\orgdiv{College of Engineering}, \orgname{Boston University}, \orgaddress{\city{Boston}, \postcode{02215}, \state{MA}, \country{United States}}}

\affil*[2]{\orgdiv{Department of Mathematics}, \orgname{University of Tennessee}, \orgaddress{\city{Knoxville}, \postcode{37996}, \state{TN}, \country{United States}}}

\affil[3]{\orgdiv{Department of Mathematics}, \orgname{University College London}, \orgaddress{\city{London}, \postcode{WC1E 6BT}, \country{UK}}}

\abstract{
We study stem--terminally differentiated (TD) lineages in small niches where demographic noise from discrete division and death events is non-negligible. Starting from a mechanistic five-channel, density-dependent CTMC (symmetric self-renewal, symmetric differentiation, asymmetric division, dedifferentiation, TD death), we derive its mean-field limit and a functional CLT, obtaining a chemical Langevin diffusion whose explicit state-dependent covariance exactly matches the CTMC’s aggregated channel-wise infinitesimal covariances.

Within this diffusion approximation we remove the dedifferentiation flux and obtain a sharp dichotomy: in subcritical regimes the stem coordinate becomes extinct asymptotically almost surely, whereas in supercritical regimes polynomial moments diverge exponentially. This identifies, at the diffusion level, a structural failure mode of strictly hierarchical lineages under demographic fluctuations and clarifies how a cyclic return flux can rescue homeostasis.

For interpretation we also derive an exact totals ODE backbone from a damage-structured transport model and obtain two steady-state constraints (ratio and equalization laws) linking compartment ratios to turnover and balancing dedifferentiation against fate bias. Numerical experiments corroborate the $\Omega^{-1/2}$ fluctuation scaling, illustrate the pathology, and contrast theorem-regime global convergence with threshold (Allee-type) behaviour outside the theorem hypotheses.
}




\maketitle

\section{Introduction}

\begin{figure}[htbp]
\centering
\begin{tikzpicture}[
  font=\small,
  >=Latex,
  node distance=7mm and 10mm,
  block/.style={draw, rounded corners=2mm, thick, align=left, inner sep=5pt, fill=gray!6},
  title/.style={font=\bfseries},
  pill/.style={draw, rounded corners=8pt, thick, inner sep=3pt, fill=white},
  arrow/.style={->, thick},
  dashbox/.style={draw, rounded corners=2mm, thick, dashed, inner sep=5pt},
  note/.style={font=\scriptsize, align=left},
  tiny/.style={font=\scriptsize, align=left}
]

\node[block, minimum width=0.47\linewidth] (ctmc) {
  {\bfseries Microscopic CTMC}\\[-0.75pt]
  $N^\Omega=(P^\Omega,W^\Omega)$,\ $X^\Omega=\Omega^{-1}N^\Omega$\\
  Feedback via $p_i(w),\lambda_P(w),\lambda_R(p)$,\ TD death $\delta_0$\\
  Five event channels with stoichiometry $\nu_e$
};

\node[pill, anchor=north east] at ([xshift=-3pt,yshift=-3pt]ctmc.north east)
{$P \rightleftarrows W$};

\node[block, right=of ctmc, minimum width=0.48\linewidth] (scaling) {
  {\bfseries Scaling limits}\\[-2pt]
  LLN: $\dot x=b(x)$\\
  FCLT: Gaussian fluctuations\\
  $A(x)=\sum_e \alpha_e\nu_e\nu_e^\top$
};

\draw[arrow] (ctmc) -- (scaling);

\node[block, below=of scaling, minimum width=0.48\linewidth] (cle) {
  {\bfseries Diffusion approximation}\\[-0.75pt]
  $dX=b(X)\,dt+\Omega^{-1/2}G(X)\,dB$\\
  Channel-wise inherited noise
};

\draw[arrow] (scaling) -- (cle);

\node[block, below=of cle, minimum width=0.48\linewidth] (analysis) {
  {\bfseries Stochastic analysis}\\[-2pt]
  Positivity; non-explosion; moment bounds
};

\draw[arrow] (cle) -- (analysis);

\node[block, below=of analysis, minimum width=0.48\linewidth] (path) {
  {\bfseries No-dedifferentiation pathology}\\[-2pt]
  $\lambda_R\equiv0$ leads to extinction or moment blow-up
};

\draw[arrow] (analysis) -- (path);

\node[block, below=of ctmc, minimum width=0.47\linewidth] (backbone) {
  {\bfseries Deterministic backbone}\\[-2pt]
  Structured PDE $\Rightarrow$ totals ODE
};

\draw[arrow] (ctmc) -- (backbone);

\node[block, below=of backbone, minimum width=0.47\linewidth] (laws) {
  {\bfseries Steady-state laws}\\[-2pt]
  Ratio and equalization constraints
};

\draw[arrow] (backbone) -- (laws);

\node[dashbox, fit=(ctmc)(scaling)(path)(laws), inner sep=6pt] {};

\end{tikzpicture}

\caption{\textbf{Multiscale roadmap.}
Event-level lineage dynamics induce deterministic and stochastic limits whose structure explains both homeostatic constraints and diffusion-level pathologies.}
\label{fig:mechanism_overview_full}
\end{figure}
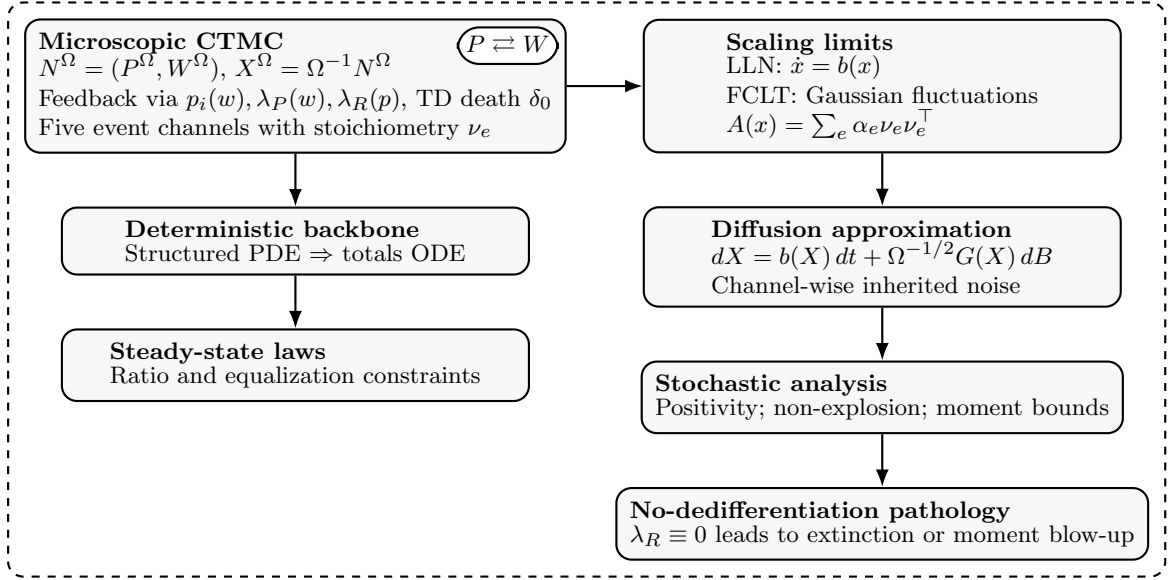

The maintenance of adult tissue homeostasis requires a balance between stem cell self-renewal, differentiation, and cell loss. 
Stem cells sustain this balance by continuously producing committed progenitors and terminally differentiated (TD) cells, forming a hierarchical differentiation structure within lineages \citep{siebert_stem_2019,bai_balance_2021,ashcroft_constrained_2022,halim_recent_2020,cheng_new_2020}. Lineage architectures also incorporate feedback mechanisms that tune cell production and help tissues respond robustly to environmental perturbations \citep{lander_cell_2009,dray_dynamic_2021,navarro_feedback_2024,hannezo_mechanochemical_2019}.
Classical formulations of the stem cell hypothesis posit a rigid, unidirectional hierarchy in which stem cells generate committed progenitors that ultimately produce TD cells \citep{till_direct_1961,potten_stem_1990,dua_stem_2003,visvader_tissue-specific_2016}. However, a growing body of experimental evidence indicates that lineage commitment can be reversible: under injury or stress, committed progenitors may regain stem-like properties through dedifferentiation \citep{tata_dedifferentiation_2013,tetteh_replacement_2016,blanpain_plasticity_2014}.
This plasticity not only supports regeneration but can also underlie oncogenic conversion and resistance to therapy \citep{haddadin_stem_2025,malta_machine_2018,shibue_emt_2017,masciale_molecular_2024,li_how_2020}.
In cancer, related plasticity mechanisms are frequently associated with epithelial-to-mesenchymal transition (EMT), in which non-stem cancer cells acquire stem-like traits \citep{jolly_implications_2015,chaffer_perspective_2011}. These observations motivate mathematical descriptions in which lineage topology is not acyclic but instead contains a dedifferentiation flux from differentiated to stem-like states.

A range of models has been proposed to interrogate the consequences of dedifferentiation in tissue maintenance and tumor progression \citep{jilkine_mathematical_2019}. Existing approaches include multi-compartment ordinary differential equation (ODE) models for cancer stem/non-stem dynamics and treatment response \citep{fischer_tumoural_2023}, stochastic hierarchical models assessing the impact of noise on cell number fluctuations \citep{wang_effect_2022}, hybrid stochastic--deterministic models for mutation acquisition and initiation \citep{jilkine_effect_2014}, and feedback-regulated evolutionary models for fixation and clonal dynamics \citep{wodarz_effect_2018,mahdipour-shirayeh_phenotypic_2017}. Additional work has addressed invasion and transient overshoot phenomena in multi-phenotypic populations \citep{zhou_invasion_2019,chen_overshoot_2016,zhou_multi-phenotypic_2014,zhou_population_2013} and the role of dedifferentiation-related factors in therapy resistance \citep{rhodes_mathematical_2016}. Despite these advances, a central mathematical issue remains: introducing a dedifferentiation flux converts the lineage graph from a directed acyclic structure into a cyclic system, thereby changing the stability architecture and the role of feedback control.

Deterministic models have been instrumental in classifying feedback mechanisms that stabilize cyclic lineages and prevent unbounded growth or extinction \citep{lander_cell_2009,marciniak-czochra_modeling_2009}. In hematopoiesis and related settings, delay differential equations and other structured formulations have been used to represent regulatory control of cell-cycle entry and fate outcomes \citep{mackey_unified_1978,adimy_global_2003,lei_multistability_2011}, and subsequent analyses have compared feedback topologies that act on division rates or self-renewal probabilities \citep{marciniak-czochra_modeling_2009,nakata_stability_2012}. To model continuous maturation variables beyond compartment counts, McKendrick--von~Foerster type partial differential equations (PDEs) provide a rigorous transport framework for structured densities (e.g., age, maturity, or accumulated damage), including renewal-type boundary fluxes generated by division \citep{foutel-rodier_individual-based_2022,lorenzi_phenotype_2025,argasinski_towards_2021,meleard_trait_2009}.

While deterministic descriptions capture average behavior, many stem cell niches operate at small effective population sizes, where demographic noise induced by discrete division and death events is non-negligible \citep{zhdanov_simulation_2008}. Such stochasticity underlies neutral drift phenomena observed in epithelial tissues \citep{simons_strategies_2011}. A natural microscopic description is therefore a master equation or continuous-time Markov chain (CTMC) for compartment counts \citep{gillespie_exact_1977,gardiner_handbook_2002}. However, direct analysis at the master-equation level is often intractable, motivating diffusion approximations that retain mechanistic interpretability while enabling rigorous analysis.

In this paper we develop a multiscale framework that links microscopic lineage events to macroscopic tissue dynamics, with the stochastic, event-based description as the primary starting point. Our main contributions are as follows.

\noindent\textbf{(i) Mechanism-consistent diffusion approximation from a density-dependent CTMC.}
We formulate an individual-based CTMC with five elementary event channels: symmetric self-renewal, symmetric differentiation, asymmetric division, dedifferentiation, and TD death. The event intensities are density-dependent and incorporate regulatory feedback through locally Lipschitz rate and fate maps. Under the standard system-size scaling, we establish the mean-field/law-of-large-numbers (LLN) limit and derive a diffusion approximation via a functional central limit theorem (FCLT) in the sense of Ethier and Kurtz \citep{ethier_markov_1986} (cf.\ the system-size expansion of van Kampen \citep{van_kampen_stochastic_2007}). The resulting chemical Langevin equation (CLE) is a system of It\^o stochastic differential equations (SDEs) with explicitly state-dependent diffusion coefficients, and its covariance matrix reproduces the exact aggregated channel-wise infinitesimal covariances, thereby avoiding ad hoc noise prescriptions.

\noindent\textbf{(ii) Well-posedness and moment bounds for the CLE.}
For the diffusion approximation we prove well-posedness up to boundary exit on the positive orthant, establish positivity preservation under the model structure, and obtain polynomial moment bounds using Lyapunov-type estimates under natural boundedness and regularity assumptions on the feedback maps. These results provide a mathematically controlled stochastic framework for analyzing finite-size lineage dynamics with mechanistically derived noise.

\noindent\textbf{(iii) No-dedifferentiation pathology in the diffusion approximation.}
We isolate the structural role of lineage plasticity by removing the dedifferentiation flux (setting the dedifferentiation rate identically zero) in the diffusion model. In this no-dedifferentiation regime we prove a dichotomy: subcritical parameter regimes yield almost sure asymptotic extinction of the stem compartment, whereas supercritical regimes lead to exponential divergence of polynomial moments. This pathology clarifies, at the level of the diffusion approximation, how the dedifferentiation flux can act as a rescue mechanism against demographic fluctuations in finite niches.

\paragraph{Scope of the stochastic pathology}
Our no-dedifferentiation pathology is established at the level of the diffusion approximation (CLE) derived from the density-dependent CTMC. Accordingly, the extinction and moment-divergence statements should be interpreted as diffusion-level structural phenomena capturing demographic fluctuations for intermediate-to-large effective system sizes; boundary-layer and discreteness effects of the underlying CTMC near very small counts are discussed in Section~\ref{sec:discussion}.

\paragraph{Modelling layers and notation (totals versus concentrations)}
We deliberately separate two modelling layers. In Sections~\ref{sec:ctmc}--\ref{sec:noR}, we work with a density-dependent CTMC and its diffusion approximation, where regulatory feedback is expressed in terms of concentrations $x=(p,w)=\Omega^{-1}(P,W)$; this concentration dependence is required to obtain a non-degenerate system-size scaling as $\Omega\to\infty$. In Sections~\ref{sec:backbone}--\ref{sec:det_dynamics}, we introduce a deterministic backbone formulated in terms of compartment totals $(P,W)$ to derive interpretable steady-state constraints (ratio and equalization laws) and to organize deterministic phase portraits. To avoid ambiguity, we denote total-based feedback maps by $p_i^{\mathrm{tot}}(W)$ and $\lambda_{P,R}^{\mathrm{tot}}(\cdot)$, reserving $p_i(w)$ and $\lambda_{P,R}(\cdot)$ for the concentration-feedback maps used in the stochastic layer.

To connect the stochastic description to a deterministic backbone and to obtain interpretable steady-state constraints, we also formulate a damage-structured McKendrick--von~Foerster transport model for stem and TD densities and derive an exact reduction to a two-compartment ODE for total masses. At positive equilibria, this reduction yields two algebraic identities (an equalization law and a ratio law) linking fate bias and turnover parameters to compartment ratios, providing a mechanistic interpretation of the stochastic phase structure. \Cref{fig:mechanism_overview_full} depicts the multiscale roadmap: from the CTMC level, to CLE formulation via LLN limit and the FCLT, to a deterministic damage-structured PDE model, which admits exact reduction to a two-compartment ODE.

The remainder of the paper is organized as follows. \Cref{sec:ctmc} formulates the microscopic lineage model as a density-dependent CTMC, defining the five elementary event channels and the system-size scaling. \Cref{sec:diffusion} establishes the mean-field limit and derives the mechanism-consistent chemical Langevin diffusion approximation via a FCLT. \Cref{sec:stochastic_analysis} provides the stochastic analysis of this diffusion, proving well-posedness, non-explosion, and polynomial moment bounds. \Cref{sec:noR} isolates the structural role of plasticity, proving the no-dedifferentiation pathology: a dichotomy between asymptotic extinction and moment divergence when the dedifferentiation flux is removed. To interpret these results, \Cref{sec:backbone} introduces a damage-structured transport model, deriving an exact reduction to a deterministic ODE backbone and establishing algebraic steady-state laws (ratio and equalization identities). \Cref{sec:det_dynamics} analyzes the phase structure of this backbone, contrasting regimes of global convergence with those exhibiting bistability and Allee-type thresholds. \Cref{sec:numerics} presents numerical experiments validating the fluctuation scaling, the pathology, and the deterministic phase portraits. We conclude in \Cref{sec:discussion} with biological implications and directions for further work.

\section{Microscopic lineage model: a density-dependent CTMC}
\label{sec:ctmc}

Deterministic compartment models provide useful mean-field descriptions of lineage turnover, but they do not encode the discreteness of division and death events that generates demographic noise in small niches. To obtain a mechanism-consistent stochastic description, we formulate an individual-based CTMC for the stem and TD compartment counts. The model is specified by a finite set of elementary event channels together with density-dependent intensities. This structure places the process in the classical framework of density-dependent population processes and provides a direct route to both the mean-field limit and the diffusion approximation developed in the next section.

\subsection{State variables and system-size scaling}
\label{subsec:ctmc_scaling}

Fix a system-size parameter $\Omega\in\mathbb{N}$, interpreted as an effective tissue volume or scaling parameter. Let
\[
N^\Omega(t)\coloneqq\big(P^\Omega(t),W^\Omega(t)\big)\in\mathbb{Z}_+^2
\]
denote the integer-valued count process, where $P^\Omega(t)$ is the number of stem cells and $W^\Omega(t)$ is the number of TD cells at time $t$. We work with the associated concentration process
\[
X^\Omega(t)\coloneqq \frac{1}{\Omega}N^\Omega(t)=\big(p^\Omega(t),w^\Omega(t)\big)\in \Omega^{-1}\mathbb{Z}_+^2.
\]
Throughout \Cref{sec:ctmc,sec:diffusion,sec:stochastic_analysis}, regulatory feedback acts on concentrations. Accordingly, we fix locally Lipschitz maps
\[
p_1,p_2:[0,\infty)\to[0,1],\qquad 
\lambda_P:[0,\infty)\to(0,\infty),\qquad 
\lambda_R:[0,\infty)\to[0,\infty),
\]
and set $p_3(w)\coloneqq 1-p_1(w)-p_2(w)\ge 0$ for all $w\ge 0$. Here $p_1(w)$ and $p_2(w)$ are probabilities of symmetric self-renewal and symmetric differentiation regulated by TD concentration $w$, $p_3(w)$ is the probability of asymmetric division, $\lambda_P(w)$ is the stem division rate regulated by TD concentration, and $\lambda_R(p)$ is the dedifferentiation rate regulated by stem concentration. The constant $\delta_0>0$ denotes the TD death rate.

\subsection{Elementary event channels and stoichiometry}
\label{subsec:ctmc_events}

We model lineage turnover through five elementary event channels. Writing reactions in terms of the count variables $(P,W)$ and listing the corresponding stoichiometric increments $\nu_e\in\mathbb{Z}^2$, we have
\begin{align*}
\mathrm{(SR)}\quad &P \longrightarrow P+1,
&&\nu_{\mathrm{SR}}=(+1,0),\\[4pt]
\mathrm{(SD)}\quad &P \longrightarrow P-1,\;\; W \longrightarrow W+2,
&&\nu_{\mathrm{SD}}=(-1,+2),\\[4pt]
\mathrm{(ASD)}\quad &W \longrightarrow W+1,
&&\nu_{\mathrm{ASD}}=(0,+1),\\[4pt]
\mathrm{(R)}\quad &P \longrightarrow P+1,\;\; W \longrightarrow W-1,
&&\nu_{\mathrm{R}}=(+1,-1),\\[4pt]
\mathrm{(D)}\quad &W \longrightarrow W-1,
&&\nu_{\mathrm{D}}=(0,-1).
\end{align*}
These channels represent, respectively, symmetric self-renewal (net gain of one stem cell), symmetric differentiation (a stem cell is replaced by two TD cells), asymmetric division (net gain of one TD cell), dedifferentiation (conversion of one TD cell into one stem cell), and TD death (loss of one TD cell).

\subsection{Propensity densities and density-dependent intensities}
\label{subsec:ctmc_intensities}

For $x=(p,w)\in\mathbb{R}_+^2$, define the propensity densities
\begin{equation}\label{eq:propensity_densities_ctmc}
\begin{aligned}
\alpha_{\mathrm{SR}}(p,w)  &\coloneqq p_1(w)\,\lambda_P(w)\,p,\\
\alpha_{\mathrm{SD}}(p,w)  &\coloneqq p_2(w)\,\lambda_P(w)\,p,\\
\alpha_{\mathrm{ASD}}(p,w) &\coloneqq p_3(w)\,\lambda_P(w)\,p,\\
\alpha_{\mathrm{R}}(p,w)   &\coloneqq \lambda_R(p)\,w,\\
\alpha_{\mathrm{D}}(p,w)   &\coloneqq \delta_0\,w.
\end{aligned}
\end{equation}
The CTMC is specified by the density-dependent intensities
\begin{equation}\label{eq:density_dependent_intensities_ctmc}
a_e^\Omega(N)\coloneqq \Omega\,\alpha_e\!\left(\frac{N}{\Omega}\right),
\qquad N\in\mathbb{Z}_+^2,\quad e\in\{\mathrm{SR},\mathrm{SD},\mathrm{ASD},\mathrm{R},\mathrm{D}\}.
\end{equation}
Equivalently, conditional on $N^\Omega(t)=N$, event $e$ occurs at rate $a_e^\Omega(N)$ and the state jumps to $N+\nu_e$. By construction, the intensities depend on the counts only through the concentrations $N/\Omega$, placing the model in the standard density-dependent class.

We record a basic regularity property that will be used repeatedly. Since the feedback maps are locally Lipschitz, each $\alpha_e$ is locally Lipschitz on $\mathbb{R}_+^2$ and locally bounded. In particular, for every compact set $K\subset\mathbb{R}_+^2$,
\begin{equation}\label{eq:local_growth_ctmc}
\sup_{x\in K}\sum_{e}\alpha_e(x)\,|\nu_e|^2<\infty.
\end{equation}
This local growth bound is the standard hypothesis ensuring that the mean-field and diffusion limits apply to the process.

\subsection{Poisson time-change representation}
\label{subsec:ctmc_time_change}

Let $\{Y_e\}_e$ be independent unit-rate Poisson processes. The count process admits the classical time-change representation
\begin{equation}\label{eq:time_change_ctmc}
N^\Omega(t)=N^\Omega(0)+\sum_{e}\nu_e\,
Y_e\!\left(\int_0^t a_e^\Omega\!\big(N^\Omega(s)\big)\,ds\right).
\end{equation}
Dividing by $\Omega$ yields the corresponding representation for the concentration process
\begin{equation}\label{eq:scaled_time_change_ctmc}
X^\Omega(t)=X^\Omega(0)+\sum_{e}\nu_e\,\Omega^{-1}\,
Y_e\!\left(\Omega\int_0^t \alpha_e\!\big(X^\Omega(s)\big)\,ds\right).
\end{equation}
This form makes transparent both the LLN limit (via $\Omega^{-1}Y_e(\Omega u)\approx u$) and the central-limit scaling (via fluctuations of order $\Omega^{-1/2}$). In \Cref{sec:diffusion} we exploit \Cref{eq:scaled_time_change_ctmc} to derive the mean-field drift and the diffusion approximation.

\subsection{Drift and infinitesimal covariance}
\label{subsec:ctmc_drift_cov}

Define the drift field $b:\mathbb{R}_+^2\to\mathbb{R}^2$ by
\begin{equation}\label{eq:drift_ctmc}
b(x)\coloneqq \sum_{e}\alpha_e(x)\,\nu_e,\qquad x=(p,w).
\end{equation}
A direct substitution of \eqref{eq:propensity_densities_ctmc} and the stoichiometries yields
\begin{equation}\label{eq:drift_explicit_ctmc}
b(p,w)=
\begin{pmatrix}
  (p_1(w)-p_2(w))\lambda_P(w)\,p + \lambda_R(p)\,w \\[4pt]
  (1-p_1(w)+p_2(w))\lambda_P(w)\,p - (\delta_0+\lambda_R(p))\,w
\end{pmatrix}.
\end{equation}
The associated (instantaneous) covariance matrix is
\begin{equation}\label{eq:covariance_ctmc}
A(x)\coloneqq \sum_{e}\alpha_e(x)\,\nu_e\nu_e^\top\in\mathbb{R}^{2\times 2}.
\end{equation}
In our model, the entries of $A$ can be written explicitly as
\begin{equation}\label{eq:covariance_entries_ctmc}
\begin{aligned}
A_{11}(p,w) &= (p_1(w)+p_2(w))\lambda_P(w)\,p + \lambda_R(p)\,w,\\
A_{22}(p,w) &= (1-p_1(w)+3p_2(w))\lambda_P(w)\,p + (\lambda_R(p)+\delta_0)\,w,\\
A_{12}(p,w) &= A_{21}(p,w)= -2p_2(w)\lambda_P(w)\,p - \lambda_R(p)\,w.
\end{aligned}
\end{equation}
These expressions identify how each event channel contributes to coordinate-wise variances and cross-covariances.

For later use, we also fix a convenient channel-wise factorization. Let $B(t)=(B^1(t),\dots,B^5(t))^\top$ be a standard Brownian motion in $\mathbb{R}^5$. Defining
\begin{equation}\label{eq:G_matrix_ctmc}
G(p,w)=
\begin{pmatrix}
\sqrt{\alpha_{\mathrm{SR}}(p,w)} & -\sqrt{\alpha_{\mathrm{SD}}(p,w)} & 0 & \sqrt{\alpha_{\mathrm{R}}(p,w)} & 0\\[4pt]
0 & 2\sqrt{\alpha_{\mathrm{SD}}(p,w)} & \sqrt{\alpha_{\mathrm{ASD}}(p,w)} & -\sqrt{\alpha_{\mathrm{R}}(p,w)} & -\sqrt{\alpha_{\mathrm{D}}(p,w)}
\end{pmatrix},
\end{equation}
we have $A(p,w)=G(p,w)G(p,w)^\top$, so $A$ is positive semidefinite by construction. This explicit factorization will be used in the diffusion approximation and in the chemical Langevin formulation.

The five-channel specification above is the minimal closure that simultaneously encodes (i) fate heterogeneity at division (symmetric self-renewal, symmetric differentiation, and asymmetric division), (ii) plasticity through a dedifferentiation flux from TD to stem cells (dedifferentiation), and (iii) turnover through TD death. The density-dependent intensities implement feedback in a manner compatible with system-size scaling. In particular, the drift \eqref{eq:drift_explicit_ctmc} is the mean-field backbone of the CTMC, whereas the covariance \eqref{eq:covariance_entries_ctmc} quantifies demographic noise arising from the same elementary channels. This mechanism-consistent pairing of drift and covariance is the key feature that enables the rigorous diffusion approximation developed in the next section.

\section{Mean-field limit and diffusion approximation}
\label{sec:diffusion}

This section develops the first rigorous link between the microscopic event model (Section~\ref{sec:ctmc}) and a tractable macroscopic stochastic description. Under the standard system-size scaling for density-dependent population processes, the concentration CTMC admits (i) a LLN limit identifying the deterministic mean-field ODE and (ii) a FCLT describing Gaussian fluctuations of order $\Omega^{-1/2}$. The latter yields a mechanism-consistent diffusion approximation (CLE) whose covariance structure matches exactly the aggregated channel-wise infinitesimal covariances derived from the CTMC. Bridging microscale stochastic event dynamics to macroscopic continuum descriptions has also been pursued in hybrid PDE--agent-based oncology models, where Gillespie-type updates and mean-field closures yield tractable continuum limits \citep{wang_analysis_2025}; here we carry out this bridge for a density-dependent CTMC via the LLN limit/FCLT, obtaining a chemical Langevin diffusion with an explicit covariance inherited from the elementary event channels.

\subsection{Mean-field limit (law of large numbers)}
\label{subsec:lln}

Fix a sequence of initial conditions $X^\Omega(0)\in\Omega^{-1}\mathbb{Z}_+^2$ such that
\[
X^\Omega(0)\to x_0\in\mathbb{R}_+^2
\qquad \text{as }\Omega\to\infty.
\]
Recall the Poisson time-change representation \eqref{eq:scaled_time_change_ctmc}:
\[
X^\Omega(t)=X^\Omega(0)+\sum_{e}\nu_e\,\Omega^{-1}\,
Y_e\!\left(\Omega\int_0^t \alpha_e\!\big(X^\Omega(s)\big)\,ds\right),
\]
where $\{Y_e\}_e$ are independent unit-rate Poisson processes and the propensity densities $\alpha_e$ are given in \eqref{eq:propensity_densities_ctmc}. Writing $\tilde Y_e(u)\coloneqq Y_e(u)-u$ and using the drift $b(x)=\sum_e \alpha_e(x)\nu_e$ from \eqref{eq:drift_ctmc}, we may rewrite
\begin{equation}\label{eq:lln_decomposition}
X^\Omega(t)=X^\Omega(0)+M^\Omega(t)+\int_0^t b\big(X^\Omega(s)\big)\,ds,
\end{equation}
with the martingale term
\begin{equation}\label{eq:martingale_term_lln}
M^\Omega(t)\coloneqq \sum_e \nu_e\,\Omega^{-1}\,
\tilde Y_e\!\left(\Omega\int_0^t \alpha_e\!\big(X^\Omega(s)\big)\,ds\right).
\end{equation}
The drift field $b$ is locally Lipschitz on $\mathbb{R}_+^2$ under the standing local Lipschitz assumptions on the feedback maps (\Cref{sec:ctmc}).

We record a standard LLN estimate for Poisson processes; a proof can be given using exponential moment bounds and a maximal inequality, and is omitted here for brevity (see \Cref{app:poisson_lln}).

\begin{lemma}[Uniform LLN for Poisson processes]\label{lem:poisson_lln}
Let $Y$ be a unit-rate Poisson process. Then for each $u_0>0$,
\[
\lim_{\Omega\to \infty}\ \sup_{0\le u\le u_0}\left|\Omega^{-1}Y(\Omega u)-u\right|=0,
\qquad \text{a.s.}
\]
Equivalently, with $\tilde Y(u)=Y(u)-u$,
\[
\lim_{\Omega\to \infty}\ \sup_{0\le u\le u_0}\left|\Omega^{-1}\tilde Y(\Omega u)\right|=0,
\qquad \text{a.s.}
\]
\end{lemma}

Define the mean-field trajectory $x(t)$ as the unique solution of the deterministic ODE
\begin{equation}\label{eq:mean_field_ode}
x(t)=x_0+\int_0^t b\big(x(s)\big)\,ds,
\qquad t\ge 0,
\end{equation}
where $b$ is given explicitly by \Cref{eq:drift_explicit_ctmc}. The following theorem states the mean-field convergence of the CTMC.

\begin{theorem}[Mean-field limit]\label{thm:mean_field_limit}
Assume that the propensity densities $\alpha_e$ are locally Lipschitz and satisfy the local growth bound \eqref{eq:local_growth_ctmc}. Let $X^\Omega$ be the concentration process associated with the density-dependent CTMC in \Cref{sec:ctmc}, and let $x$ solve \eqref{eq:mean_field_ode}. If $X^\Omega(0)\to x_0$, then for every $T>0$,
\[
\lim_{\Omega\to\infty}\ \sup_{0\le t\le T}\big\|X^\Omega(t)-x(t)\big\|=0,
\qquad \text{a.s.}
\]
\end{theorem}

\begin{proof}[Proof sketch]
Subtract \eqref{eq:mean_field_ode} from \eqref{eq:lln_decomposition} to obtain
\[
X^\Omega(t)-x(t)=X^\Omega(0)-x_0+M^\Omega(t)+\int_0^t\big(b(X^\Omega(s))-b(x(s))\big)\,ds.
\]
By local Lipschitz continuity of $b$ on compact subsets and Gr\"onwall's inequality, it remains to control $\sup_{t\le T}\|M^\Omega(t)\|$. Lemma~\ref{lem:poisson_lln} implies that the scaled centered Poisson terms in \eqref{eq:martingale_term_lln} vanish uniformly on $[0,T]$ as $\Omega\to\infty$, yielding the claim.
\end{proof}

\subsection{Functional central limit theorem and diffusion approximation}
\label{subsec:fclt}

The mean-field limit captures only the deterministic drift. To quantify intrinsic demographic fluctuations we consider the central-limit scaling around the deterministic solution. Define the fluctuation process
\begin{equation}\label{eq:fluctuation_process}
V^\Omega(t)\coloneqq \sqrt{\Omega}\,\big(X^\Omega(t)-x(t)\big),
\qquad t\ge 0.
\end{equation}
Let $A(x)$ be the instantaneous covariance matrix associated with the elementary channels,
\begin{equation}\label{eq:covariance_matrix_fclt}
A(x)\coloneqq \sum_e \alpha_e(x)\,\nu_e\nu_e^\top,
\end{equation}
which in our model has the explicit entries \eqref{eq:covariance_entries_ctmc}. The standard FCLT for density-dependent Markov chains (Ethier--Kurtz) yields a Gaussian limit for $V^\Omega$, stated below in a form tailored to our setting.

\begin{theorem}[FCLT / diffusion approximation]\label{thm:fclt}
Assume that the propensity densities $\alpha_e$ are locally Lipschitz and satisfy the local growth bound \eqref{eq:local_growth_ctmc}, and that the drift field $b$ is continuously differentiable on $\mathbb{R}_+^2$. Let $x$ solve the mean-field ODE \eqref{eq:mean_field_ode} and define $V^\Omega$ by \eqref{eq:fluctuation_process}. If $V^\Omega(0)\Rightarrow V(0)$ as $\Omega\to\infty$, then $V^\Omega \Rightarrow V$ in distribution, where $V$ is the (time-inhomogeneous) Gaussian process solving
\begin{equation}\label{eq:linear_fluctuation_sde}
V(t)=V(0)+\int_0^t \nabla b(x(s))\,V(s)\,ds + \sum_e \nu_e\,W_e\!\left(\int_0^t \alpha_e(x(s))\,ds\right),
\end{equation}
with $\{W_e\}_e$ independent standard Brownian motions.
Equivalently, $V$ has instantaneous covariance
\[
\mathbb{E}\!\left[dV(t)\,dV(t)^\top\right]=A(x(t))\,dt.
\]
\end{theorem}

The limiting fluctuation equation \eqref{eq:linear_fluctuation_sde} induces the familiar chemical Langevin approximation for $X^\Omega$ itself. Let $G:\mathbb{R}_+^2\to\mathbb{R}^{2\times 5}$ be any measurable matrix square root of $A$, i.e.\ $G(x)G(x)^\top=A(x)$; see \citep{allen_introduction_2011} for the formulation of such $G$ and alternative choices. Using the explicit factorization \eqref{eq:G_matrix_ctmc} we obtain the It\^o SDE
\begin{equation}\label{eq:cle_sde}
dX(t)=b(X(t))\,dt+\Omega^{-1/2}G(X(t))\,dB(t),
\qquad X(0)=X_0\in(0,\infty)^2,
\end{equation}
where $B$ is a standard Brownian motion in $\mathbb{R}^5$. In particular, the diffusion coefficient in \eqref{eq:cle_sde} is mechanism-consistent: the infinitesimal covariance of the SDE satisfies
\[
G(x)G(x)^\top=A(x)=\sum_e \alpha_e(x)\,\nu_e\nu_e^\top,
\]
which exactly matches the aggregated channel-wise infinitesimal covariance of the underlying CTMC.

For completeness we record the explicit drift and covariance for the present model. The drift $b$ is given by \eqref{eq:drift_explicit_ctmc}, and the covariance entries are
\begin{equation}\label{eq:diffusion_entries_repeat}
\begin{aligned}
A_{11}(p,w) &= (p_1(w)+p_2(w))\lambda_P(w)\,p + \lambda_R(p)\,w,\\
A_{22}(p,w) &= (1-p_1(w)+3p_2(w))\lambda_P(w)\,p + (\lambda_R(p)+\delta_0)\,w,\\
A_{12}(p,w) &= A_{21}(p,w)= -2p_2(w)\lambda_P(w)\,p - \lambda_R(p)\,w.
\end{aligned}
\end{equation}
We emphasize that \eqref{eq:cle_sde} preserves the deterministic mean-field limit in the sense that $\mathbb{E}[X(t)]$ follows the drift $b$ to leading order as $\Omega\to\infty$, while capturing intrinsic fluctuations of size $\mathcal{O}(\Omega^{-1/2})$ with a covariance structure inherited from the discrete event channels.

\section{Stochastic analysis of the chemical Langevin diffusion}
\label{sec:stochastic_analysis}

This section studies the chemical Langevin diffusion obtained in \Cref{sec:diffusion}. Our objectives are threefold: (i) to specify a well-posedness framework on the positive orthant, where the diffusion is state-dependent and may degenerate near the axes; (ii) to provide a non-explosion criterion ensuring that trajectories neither blow up nor reach the boundary in finite time; and (iii) to establish polynomial moment bounds. These results provide the technical foundation for the no-dedifferentiation pathology proved in \Cref{sec:noR}.
There are different approaches for demographic noise in population dynamics and epidemic models \citep{mandal_stochastic_2014,dos_santos_discreteness_2013,constable_demographic_2016,weissmann_simulation_2018}.
Our diffusion-level analysis follows a standard SDE toolkit for models driven by pure demographic noise, combining well-posedness on the positive orthant with Lyapunov control and moment bounds; see, for instance, some frameworks developed for stochastic predator--prey systems \citep{abundo_stochastic_1991,wang_analysis_2025-1}. 

\subsection{The diffusion model on the positive orthant}
\label{subsec:cle_model}

Fix $\Omega\in\mathbb{N}$. We consider the It\^o SDE on concentrations
\begin{equation}\label{eq:cle_main}
dX(t)=b(X(t))\,dt+\Omega^{-1/2}G(X(t))\,dB(t),
\qquad X(0)=X_0\in (0,\infty)^2,
\end{equation}
where $B$ is a standard Brownian motion in $\mathbb{R}^5$, the drift $b$ is given by \eqref{eq:drift_explicit_ctmc}, and the volatility factor $G$ is given by \eqref{eq:G_matrix_ctmc}. By construction,
\[
A(x)\coloneqq G(x)G(x)^\top=\sum_e \alpha_e(x)\,\nu_e\nu_e^\top
\]
is the mechanism-consistent covariance matrix of the diffusion approximation (\Cref{subsec:fclt}). We work on the open positive orthant
\[
U\coloneqq (0,\infty)^2.
\]
The generator of \eqref{eq:cle_main} is, for $f\in C_b^2(U)$,
\begin{equation}\label{eq:generator_main}
(\mathcal{L}f)(x)=b(x)\cdot \nabla f(x)+\frac{1}{2\Omega}\mathrm{Tr}\big(A(x)\nabla^2 f(x)\big),
\qquad x\in U.
\end{equation}

\subsection{Stopped strong solutions}
\label{subsec:stopped_solutions}

Since the coefficients in \eqref{eq:cle_main} are locally Lipschitz on $U$ (\Cref{app:ek_verification}), classical SDE theory yields existence and pathwise uniqueness up to the first exit from compact subsets of $U$. To formalize this, for $m\in\mathbb{N}$ define the compact exhaustion
\[
K_m\coloneqq [m^{-1},m]^2\subset U,
\]
and the corresponding exit time
\begin{equation}\label{eq:exit_time_def}
\tau_m\coloneqq \inf\{t\ge 0:\ X(t)\notin K_m\},
\qquad 
\tau\coloneqq \lim_{m\to\infty}\tau_m.
\end{equation}

\begin{proposition}[Existence and uniqueness up to exit]\label{prop:stopped_strong}
Assume that $b$ and $G$ are locally Lipschitz and locally bounded on $U$. Then for every $m\in\mathbb{N}$ and every initial condition $X_0\in U$, there exists a unique strong solution to \eqref{eq:cle_main} on $[0,\tau_m]$. In particular, $X(t)\in K_m$ for all $t<\tau_m$, and pathwise uniqueness holds up to $\tau_m$.
\end{proposition}

\begin{proof}[Proof sketch]
On the compact set $K_m$, local Lipschitz continuity implies that $b$ and $G$ are bounded and Lipschitz. Extend $b|_{K_m}$ and $G|_{K_m}$ to globally Lipschitz, bounded maps on $U$ via a smooth cutoff (a standard construction; see \Cref{app:cutoff_extension}). The corresponding globally Lipschitz SDE admits a unique global strong solution. By construction, this solution agrees with \eqref{eq:cle_main} up to the exit time from $K_m$, which is $\tau_m$. See \citep{karatzas_brownian_1998} for the standard construction.
\end{proof}

\subsection{A Lyapunov criterion for non-explosion}
\label{subsec:lyapunov}

Proposition~\ref{prop:stopped_strong} isolates the only possible obstructions to global existence: either the process approaches the boundary of $U$ (one coordinate tends to $0$) or it escapes to infinity. A standard way to rule out both events is to construct a Lyapunov function that diverges at $\partial U$ and at infinity and has at most linear growth under the generator.

\begin{proposition}[Lyapunov non-explosion criterion]\label{prop:lyapunov_nonexplosion}
Assume the hypotheses of Proposition~\ref{prop:stopped_strong}. Suppose there exist $V\in C^2(U;[0,\infty))$ and constants $C_0,C_1>0$ such that
\begin{equation}\label{eq:V_blowup_main}
\lim_{\substack{x\to \partial U\\ x\in U}}V(x)=\infty,
\qquad 
\lim_{\substack{\|x\|\to\infty\\ x\in U}}V(x)=\infty,
\end{equation}
and such that the generator \eqref{eq:generator_main} satisfies
\begin{equation}\label{eq:LV_drift_main}
(\mathcal{L}V)(x)\le C_0+C_1 V(x),
\qquad x\in U.
\end{equation}
Then $\tau=\infty$ almost surely. In particular, the unique strong solution to \eqref{eq:cle_main} exists for all $t\ge 0$ and does not reach $\partial U$ or explode to infinity in finite time.
\end{proposition}

\begin{proof}[Proof sketch]
Apply It\^o's formula to $V(X(t\wedge \tau_m))$, where $\tau_m$ is defined in \eqref{eq:exit_time_def}. Taking expectations and using \eqref{eq:LV_drift_main} yields
\[
\mathbb{E}\big[V(X(t\wedge \tau_m)) \big] \le V(X_0) + \int_0^t \Big(C_0+C_1 \mathbb{E}\big[V(X(s\wedge \tau_m)) \big]\Big)\,ds.
\]
Gr\"onwall's inequality implies that,
\[
\mathbb{E}\big[V(X(t\wedge \tau_m))\big]\le \big(V(X_0)+C_0 t\big)e^{C_1 t},
\qquad t\ge 0,
\]
uniformly in $m$. If $\mathbb{P}(\tau<\infty)>0$, then on $\{\tau\le t\}$ one has $X(\tau_m)\in\partial K_m$ for all $m$ sufficiently large, and \eqref{eq:V_blowup_main} implies $\inf_{x\in\partial K_m}V(x)\to\infty$, contradicting the uniform bound. Hence $\mathbb{P}(\tau<\infty)=0$.
\end{proof}

\subsection{Polynomial moment bounds}
\label{subsec:moments}

We next establish polynomial moment bounds for the diffusion, initially up to the exit time $\tau$, and globally when $\tau=\infty$ almost surely. We formulate a convenient linear-growth assumption on the coefficients; for the present model it is implied by bounded feedback maps, since the propensities are at most linear in $(p,w)$.

\begin{assumption}[Linear-growth bounds]\label{ass:linear_growth}
There exist constants $K_b,K_A\ge 0$ such that for all $x\in U$,
\begin{equation}\label{eq:linear_growth_assumption}
x\cdot b(x)\le K_b\,(1+\|x\|^2),
\qquad 
\mathrm{Tr}\,A(x)\le K_A\,(1+\|x\|^2),
\end{equation}
where $A(x)=G(x)G(x)^\top$.
\end{assumption}

\begin{remark}
Assumption~\ref{ass:linear_growth} holds for the present diffusion when $p_1,p_2,\lambda_P,\lambda_R$ are bounded on $[0,\infty)$: each propensity $\alpha_e(p,w)$ is then at most linear in $(p,w)$, and $b$ and $A$ inherit the corresponding growth bounds.
\end{remark}

\begin{theorem}[Polynomial moment bounds up to exit]\label{thm:moment_bounds}
Assume Assumption~\ref{ass:linear_growth}. Let $q\ge 4$ and set $V_q(x)\coloneqq \|x\|^q$ for $x\in U$. Then there exists a constant $C_q>0$ (depending only on $q$, $K_b$, $K_A$, and $\Omega$) such that for every $m\in\mathbb{N}$ and all $t\ge 0$,
\begin{equation}\label{eq:moment_bound_stopped}
\mathbb{E}\big[\|X(t\wedge \tau_m)\|^q\big]\le \big(\mathbb{E}\|X_0\|^q+C_q t\big)\,e^{C_q t}.
\end{equation}
Consequently, letting $m\to\infty$ yields
\begin{equation}\label{eq:moment_bound_tau}
\mathbb{E}\big[\|X(t\wedge \tau)\|^q\big]\le \big(\mathbb{E}\|X_0\|^q+C_q t\big)\,e^{C_q t},
\qquad t\ge 0.
\end{equation}
If in addition $\tau=\infty$ almost surely (e.g.\ by Proposition~\ref{prop:lyapunov_nonexplosion}), then \eqref{eq:moment_bound_tau} holds with $t\wedge\tau=t$ for all $t\ge 0$.
\end{theorem}

\begin{proof}[Proof sketch]
Fix $m\in\mathbb{N}$ and apply It\^o's formula to $V_q(X(t\wedge\tau_m))$. Using
\[
\nabla V_q(x)=q\|x\|^{q-2}x,\qquad 
\nabla^2 V_q(x)=q\|x\|^{q-2}I+q(q-2)\|x\|^{q-4}xx^\top,
\]
we obtain
\[
(\mathcal{L}V_q)(x)
= q\|x\|^{q-2}\,x\cdot b(x)
+\frac{q}{2\Omega}\|x\|^{q-2}\mathrm{Tr}\,A(x)
+\frac{q(q-2)}{2\Omega}\|x\|^{q-4}\,x^\top A(x)x.
\]
Since $A(x)$ is positive semidefinite, $x^\top A(x)x\le \|x\|^2\mathrm{Tr}\,A(x)$ (see \Cref{app:moments_details}), hence
\[
(\mathcal{L}V_q)(x)
\le q\|x\|^{q-2}\,x\cdot b(x)
+\frac{q(q-1)}{2\Omega}\|x\|^{q-2}\mathrm{Tr}\,A(x).
\]
Applying Assumption~\ref{ass:linear_growth} and the elementary bound $\|x\|^{q-2}\le 1+\|x\|^q$ yields
\[
(\mathcal{L}V_q)(x)\le C_q\big(1+\|x\|^q\big)
\]
for a constant $C_q$ depending on $q,K_b,K_A,\Omega$. Integrating the It\^o identity and applying Gr\"onwall's inequality gives \eqref{eq:moment_bound_stopped}. The bound \eqref{eq:moment_bound_tau} follows by letting $m\to\infty$ and using uniform integrability, which is ensured by \eqref{eq:moment_bound_stopped}.
\end{proof}

\section{No-dedifferentiation pathology in the diffusion approximation}
\label{sec:noR}

This section isolates the structural role of lineage plasticity by removing the dedifferentiation flux from the chemical Langevin diffusion. Concretely, we set the dedifferentiation rate identically zero and study the resulting diffusion on the positive orthant. The main result is a dichotomy: in subcritical regimes the stem compartment undergoes almost sure asymptotic extinction, while in supercritical regimes polynomial moments grow at least exponentially. These two behaviors constitute a \emph{no-dedifferentiation pathology} in the diffusion approximation, in the sense that the model cannot sustain a stationary homeostatic stem population under demographic fluctuations without a dedifferentiation mechanism.

\smallskip
\noindent\textbf{Scope.} All results in this section are formulated at the level of the diffusion approximation (CLE). The relationship to the underlying CTMC near the boundary---including discreteness effects and potential boundary-layer phenomena associated with diffusion degeneracy---is discussed in \Cref{sec:discussion}.

\subsection{The no-dedifferentiation diffusion and notions of population loss}
\label{subsec:noR_statement}

Fix $\Omega\in\mathbb{N}$ and consider the diffusion \eqref{eq:cle_main} with dedifferentiation removed, i.e.\ with $\lambda_R\equiv 0$. Writing $X(t)=(p(t),w(t))$, we obtain
\begin{equation}\label{eq:cle_noR}
dX(t)=b_0(X(t))\,dt+\Omega^{-1/2}G_0(X(t))\,dB(t),
\qquad X(0)=X_0\in(0,\infty)^2,
\end{equation}
where $b_0$ and $G_0$ denote the drift and volatility obtained from \eqref{eq:drift_explicit_ctmc} and \eqref{eq:G_matrix_ctmc} by setting $\lambda_R\equiv 0$. All statements below are understood for the stopped strong solution up to the exit time $\tau$ from $U=(0,\infty)^2$ as constructed in Section~\ref{sec:stochastic_analysis}.

A key point is to distinguish two notions of population loss that are often conflated in the diffusion literature.

\begin{definition}[Finite-time boundary hitting]\label{def:hitting}
Let $X(t)=(p(t),w(t))$ be a (possibly stopped) solution of \eqref{eq:cle_noR}. We say that the stem coordinate hits the boundary in finite time if
\[
\tau_{p=0}\coloneqq \inf\{t\ge 0:\ p(t)=0\}<\infty.
\]
Analogously, $\tau_{w=0}\coloneqq \inf\{t\ge 0:\ w(t)=0\}<\infty$ denotes finite-time hitting of the TD boundary.
\end{definition}

\begin{definition}[Asymptotic extinction]\label{def:asymptotic_extinction}
We say that the stem coordinate exhibits asymptotic extinction if
\[
p(t)\longrightarrow 0 \quad \text{as } t\to\infty \quad \text{almost surely on the event } \{\tau=\infty\}.
\]
If $\mathbb{P}(\tau<\infty)>0$, the preceding notion is interpreted as a statement on the non-explosion event $\{\tau=\infty\}$.
\end{definition}

The distinction is essential. Finite-time boundary hitting is a boundary attainability property of the diffusion (and depends sensitively on the degeneracy structure of $G_0$ near the axes), whereas asymptotic extinction concerns long-time decay and may occur even when the boundary is not reached at any finite time.

\subsection{Pathology without dedifferentiation}
\label{subsec:noR_pathology}

We now state and prove the no-dedifferentiation pathology. The result is formulated under uniform bounds on the division rate and a uniform bias condition on the net stem drift $p_1-p_2$.

\begin{proposition}[No-dedifferentiation pathology]\label{prop:no_dedifferentiation_main}
Assume that $\lambda_P$ is bounded and strictly positive on $[0,\infty)$, i.e.\ there exist finite constants
\[
0<\lambda_{P,\min}\le \lambda_P(w)\le \lambda_{P,\max}<\infty,\qquad \forall\,w\ge 0.
\]
Consider the no-dedifferentiation diffusion \eqref{eq:cle_noR}. Then the following two regimes occur.
\begin{enumerate}[label=\textnormal{(\roman*)}]
\item \textnormal{(Subcritical regime: asymptotic extinction).}
Assume $\inf_{w\ge 0}\,\big(p_1(w)+p_2(w)\big)\ge p_{\min}>0$ and
\begin{equation}\label{eq:subcritical_main}
\sup_{w\ge 0}\big(p_1(w)-p_2(w)\big)\le -\gamma<0
\end{equation}
for some $\gamma>0$. Let $\kappa\coloneqq \gamma\lambda_{P,\min}>0$. Then
\[
p(t)\longrightarrow 0
\qquad \text{almost surely as } t\to\infty \text{ on } \{\tau=\infty\}.
\]

\item \textnormal{(Supercritical regime: exponential moment divergence).}
Assume
\begin{equation}\label{eq:supercritical_main}
\inf_{w\ge 0}\,\big(p_1(w)-p_2(w)\big)\ge \gamma>0
\end{equation}
for some $\gamma>0$, and set $\kappa\coloneqq \gamma\lambda_{P,\min}>0$. Then for every $q\ge 1$ and all $t\ge 0$,
\begin{equation}\label{eq:dichotomy}
\mathbb{E}\big[p(t\wedge\tau)^q\big]\ \ge\ \mathbb{E}\big[p(0)^q\big]\exp(q\kappa t).
\end{equation}
Consequently one of the following alternatives holds:
\begin{enumerate}
  \item $\tau=\infty$ almost surely; or
  \item for the same $q$ at least one of the following is true:
    \[
    \mathbb{E}\big[p(\tau)^q\big]=\infty\qquad\text{or}\qquad
    \lim_{t\to\infty}\mathbb{E}\big[p(t)^q\mathbf 1_{\{\tau>t\}}\big]=\infty.
    \]
\end{enumerate}
Equivalently, if $\mathbb{E}\big[p(\tau)^q\big]<\infty$, we have the conditional lower bound
\[
\mathbb{E}\big[p(t)^q\mid \tau>t\big]\ \ge\ \frac{\mathbb{E}[p(0)^q]}{\mathbb{P}(\tau>t)}\,e^{q\kappa t} - \frac{\mathbb{E}[p(\tau)^q]}{\mathbb{P}(\tau>t)},
\]
so conditional moments on the survival event diverge asymptotically exponentially.
\end{enumerate}
\end{proposition}

\begin{proof}
We give the core constructions for the two regimes, emphasizing the supermartingale and quadratic-variation mechanisms that drive the dichotomy. The detailed calculations can be found in \Cref{app:DDS_details}.

\smallskip
\noindent\emph{(i) Subcritical regime.}
Assume that the subcritical condition
\[
\sup_{w\ge 0}\big(p_1(w)-p_2(w)\big)\le -\gamma<0
\]
holds. Define
\[
Y(t):=e^{\kappa t}p(t),\qquad \kappa:=\gamma\,\lambda_{P,\min}>0.
\]
By It\^o's formula, for all $t<\tau$,
\[
dY(t)
= e^{\kappa t}\Big((p_1(w(t))-p_2(w(t)))\lambda_P(w(t))+\kappa\Big)p(t)\,dt
+ dM(t),
\]
where $M$ is a continuous local martingale. By the subcritical assumption and the lower bound
$\lambda_P\ge \lambda_{P,\min}$, the drift term is non-positive. 
Consequently, $Y(\cdot\wedge\tau)$ is a nonnegative local supermartingale, hence a supermartingale.
Since for any $t\ge 0$, $\mathbb{E}[Y(t\wedge \tau)]\le \mathbb{E}[Y(0)]=p(0)$, Doob's martingale convergence theorem implies the almost sure existence of the limit
\[
L:=\lim_{t\to\infty}Y(t\wedge \tau)
\]
and $L\ge 0$.

We now show that $L=0$ almost surely on $\{\tau=\infty\}$. Suppose, by contradiction, that there exists a set
$E\subset\{\tau=\infty\}$ with $\mathbb{P}(E)>0$ on which $L>0$.
Fix $\omega_0\in E$ and set $L=L(\omega_0)>0$.
By convergence of $Y(t)(\omega_0)$, there exists $T(\omega_0)$ such that
\[
Y(t)(\omega_0)\ge \tfrac{1}{2}L \qquad \text{for all } t\ge T(\omega_0).
\]
Using the quadratic variation identity,
\begin{equation}\label{eq:quad_var_subcritical}
\langle M\rangle_{t}
=\frac{1}{\Omega}\int_0^{t} e^{2\kappa s}\big(p_1(w(s))+p_2(w(s))\big)\lambda_P(w(s))p(s)\,ds
\ge \frac{p_{\min}\lambda_{P,\min}}{\Omega}\int_0^{t} e^{\kappa s}Y(s)\,ds,
\end{equation}
we obtain for $t\ge T(\omega_0)$,
\[
\int_0^{t} e^{\kappa s}Y(s)(\omega_0)\,ds
\ge \frac{L}{2}\int_{T(\omega_0)}^{t} e^{\kappa s}\,ds
= \frac{L}{2\kappa}\big(e^{\kappa t}-e^{\kappa T(\omega_0)}\big)
\xrightarrow[t\to\infty]{}\infty.
\]
Hence $\langle M\rangle_t(\omega_0)\to\infty$ as $t\to\infty$.

The Doob--Meyer decomposition applying to the supermartingale $Y(\cdot\wedge \tau)$ implies that there exist a martingale $\widetilde M$ and an increasing predictable process $A$ with $A(0)=0$ such that
\[
Y(t\wedge \tau)=Y(0)+\widetilde M(t)-A(t),\qquad t\ge0.
\]
For a nonnegative supermartingale one has $\sup_{t}\mathbb{E}[A(t)]<\infty$, hence $A(t)\to A(\infty)<\infty$ almost surely. Therefore, the convergence of $Y(t\wedge \tau)$ as $t\to\infty$ ensures that the martingale component $\widetilde M(t)=Y(t\wedge \tau)-Y(0)+A(t)$ must also converge almost surely.
This contradicts the fact that a continuous local martingale with divergent quadratic variation
cannot converge; see \Cref{app:DDS_details}.
The contradiction implies that 
\[
\lim_{t\to\infty}Y(t)=0 \quad \text{a.s.\ on } \{\tau=\infty\}.
\]
Consequently,
\[
p(t)=e^{-\kappa t}Y(t)\longrightarrow 0 \qquad \text{a.s.\ on } \{\tau=\infty\},
\]
establishing almost sure extinction in the subcritical regime.

\smallskip
\noindent\emph{(ii) Supercritical regime.}
Assume \eqref{eq:supercritical_main}. The stem drift in \eqref{eq:cle_noR} satisfies
\[
(p_1(w(t))-p_2(w(t)))\lambda_P(w(t))\ge \kappa,\qquad t<\tau.
\]
Fix $q\ge 1$ and apply It\^o's formula to $p(t)^q$ (stopped at $\tau_m$). The diffusion correction term is nonnegative, so we obtain
\[
d\big(p(t)^q\big)\ \ge\ q\kappa\,p(t)^q\,dt + d\mathcal{M}_q(t),
\qquad t<\tau_m,
\]
for a local martingale $\mathcal{M}_q$. Taking expectations and using $\mathbb{E}[\mathcal{M}_q(t\wedge\tau_m)]=0$ yields
\[
\frac{d}{dt}\,\mathbb{E}\big[p(t\wedge\tau_m)^q\big]\ \ge\ q\kappa\,\mathbb{E}\big[p(t\wedge\tau_m)^q\big].
\]
Gr\"onwall's inequality gives
\[
\mathbb{E}\big[p(t\wedge\tau_m)^q\big]\ \ge\ \mathbb{E}\big[p(0)^q\big]e^{q\kappa t}.
\]
Letting $m\to\infty$ yields the stated bound for $p(t\wedge\tau)$; the passage to the limit can be justified by uniform integrability using the polynomial moment bounds of Section~\ref{subsec:moments}. 

Write
\[
m(t):=\mathbb{E}\big[p(t\wedge\tau)^q\big]
=\mathbb{E}\big[p(\tau)^q\mathbf 1_{\{\tau\le t\}}\big]+\mathbb{E}\big[p(t)^q\mathbf 1_{\{\tau>t\}}\big].
\]
By \eqref{eq:dichotomy} we have \(m(t)\to\infty\) exponentially as \(t\to\infty\). Two cases:

If \(\mathbb{E}[p(\tau)^q]=\infty\) then the second alternative (exit-moment divergence) holds and we are done.

If \(\mathbb{E}[p(\tau)^q]<\infty\) then the first term
\(\mathbb{E}\big[p(\tau)^q\mathbf 1_{\{\tau\le t\}}\big]\) is uniformly bounded in \(t\). Since \(m(t)\to\infty\), the second term must diverge:
\[
\lim_{t\to\infty}\mathbb{E}\big[p(t)^q\mathbf 1_{\{\tau>t\}}\big]=\infty,
\]
which is the stated alternative. This proves the dichotomy.

Finally, the conditional version is immediate from (1) by dividing both sides by \(\mathbb{P}(\tau>t)\) and observing when $\mathbb{E}[p(\tau)^q]<\infty$,
\[
\mathbb{E}\big[p(t\wedge\tau)^q\big]
\le \mathbb{E}\big[p(t)^q\mathbf 1_{\{\tau>t\}}\big] + \mathbb{E}\big[p(\tau)^q\mathbf 1_{\{\tau\le t\}}\big]
\le \mathbb{E}\big[p(t)^q\mathbf 1_{\{\tau>t\}}\big] + \mathbb{E}\big[p(\tau)^q\big].
\]
\end{proof}

\begin{remark}[Mechanistic interpretation]\label{rmk:noR_interpretation}
In the full model \eqref{eq:cle_main}, dedifferentiation contributes the positive term $+\lambda_R(p)w$ to the stem drift. At a deterministic positive equilibrium, this dedifferentiation flux balances the net differentiation bias through the equalization identity (Section~\ref{sec:backbone}). When $\lambda_R\equiv 0$, this balancing mechanism is structurally absent, and the diffusion approximation exhibits either asymptotic extinction (when the net stem drift is uniformly negative) or moment divergence (when it is uniformly positive). In this sense, the dedifferentiation flux provides a mechanism to mitigate demographic fluctuations in finite niches.
\end{remark}

\section{Deterministic backbone for interpretation: structured PDE reduction and steady-state laws}
\label{sec:backbone}

Sections~\ref{sec:ctmc}--\ref{sec:noR} develop a mechanism-consistent stochastic description starting from a density-dependent CTMC and its diffusion approximation. In this section we introduce a deterministic ``backbone'' that serves two interpretive purposes. First, it provides a structured (McKendrick-von Foerster or damage-structured PDE) formulation that justifies the two-compartment description at the level of total masses.
Second, it yields explicit algebraic identities at positive equilibria (the ratio and equalization laws), which clarify how fate bias and turnover parameters must balance at homeostasis.
Importantly, the deterministic results in this section are used for interpretation and parameter constraints rather than as the primary driver of the stochastic analysis.

Throughout this section, we work with \emph{total-based} feedback maps and use superscripts to avoid confusion with the concentration feedback used in Sections~\ref{sec:ctmc}--\ref{sec:noR}. Specifically, we write
\[
p_1^{\mathrm{tot}},\,p_2^{\mathrm{tot}}:[0,\infty)\to[0,1],\qquad
\lambda_P^{\mathrm{tot}}:[0,\infty)\to(0,\infty),\qquad
\lambda_R^{\mathrm{tot}}:[0,\infty)\to[0,\infty),
\]
with $p_3^{\mathrm{tot}}(W)\coloneqq 1-p_1^{\mathrm{tot}}(W)-p_2^{\mathrm{tot}}(W)\ge 0$. These maps need not coincide with the concentration-feedback maps $(p_1,p_2,\lambda_P,\lambda_R)$ used in the stochastic layer.

\subsection{A damage-structured PDE model and exact reduction to totals}
\label{subsec:pde_to_ode}

We formulate a two-compartment lineage structured by a nonnegative damage load $x\in\mathbb{R}_+$. Let $S(t,x)$ denote the stem-cell density and $T(t,x)$ the TD density at time $t\ge 0$ and damage level $x\ge 0$. Define the compartment totals
\begin{equation}\label{eq:totals_def_backbone}
P(t)\coloneqq \int_0^\infty S(t,x)\,dx,\qquad 
W(t)\coloneqq \int_0^\infty T(t,x)\,dx.
\end{equation}
Damage accumulation is represented by nonnegative transport velocities $v_S,v_T:\mathbb{R}_+\to[0,\infty)$. Regulatory feedback is assumed to act through the totals $P(t)$ and $W(t)$ via the maps
\[
\lambda_P^{\mathrm{tot}}(W),\qquad \lambda_R^{\mathrm{tot}}(P),\qquad
p_1^{\mathrm{tot}}(W),\qquad p_2^{\mathrm{tot}}(W),
\]
and we fix a constant TD death rate $\delta_0>0$. The structured PDE system is
\begin{equation}\label{eq:PDE_system_backbone}
\begin{cases}
\partial_t S(t,x) + \partial_x\!\big(v_S(x)\,S(t,x)\big)
= -\lambda_P^{\mathrm{tot}}(W(t))\,S(t,x) + \lambda_R^{\mathrm{tot}}(P(t))\,T(t,x),
\\[4pt]
\partial_t T(t,x) + \partial_x\!\big(v_T(x)\,T(t,x)\big)
= -\big(\delta_0+\lambda_R^{\mathrm{tot}}(P(t))\big)\,T(t,x),
\end{cases}
\qquad t>0,\ x>0.
\end{equation}
Division is assumed to occur at rate $\lambda_P^{\mathrm{tot}}(W(t))$ independent of damage. Conditional on division, the outcome is symmetric self-renewal with probability $p_1^{\mathrm{tot}}(W(t))$, symmetric differentiation with probability $p_2^{\mathrm{tot}}(W(t))$, and asymmetric division with probability $p_3^{\mathrm{tot}}(W(t))$. We postulate that division resets newborn cell damage to $x=0$. Under this reset-to-zero assumption, the renewal boundary conditions at $x=0$ are
\begin{equation}\label{eq:renewal_bc_backbone}
\begin{cases}
v_S(0)\,S(t,0)
= \big(1+p_1^{\mathrm{tot}}(W(t))-p_2^{\mathrm{tot}}(W(t))\big)\,\lambda_P^{\mathrm{tot}}(W(t))\,P(t),
\\[4pt]
v_T(0)\,T(t,0)
= \big(1-p_1^{\mathrm{tot}}(W(t))+p_2^{\mathrm{tot}}(W(t))\big)\,\lambda_P^{\mathrm{tot}}(W(t))\,P(t),
\end{cases}
\qquad t>0.
\end{equation}
We assume integrable initial data $S(0,\cdot)=S_0$, $T(0,\cdot)=T_0$ with $S_0,T_0\ge 0$ and impose the flux integrability condition
\begin{equation}\label{eq:flux_vanish_backbone}
v_S(\cdot)\,S(t,\cdot)\in L^1(\mathbb{R}_+),\qquad v_T(\cdot)\,T(t,\cdot)\in L^1(\mathbb{R}_+), \qquad \text{for a.e. } t>0,
\end{equation}
which implies
\begin{equation}\label{eq:vanishing_infty}
\lim_{x\to\infty}v_S(x)S(t,x)=\lim_{x\to\infty}v_T(x)T(t,x)=0
\end{equation}
along a subsequence.

We record a standard well-posedness statement; a proof in the $L^1$ semigroup / fixed-point framework is given in \Cref{app:pde_wellposed}. For the explicit construction and detailed calculations, we refer to \citep{liang_global_2025}.

\begin{theorem}[Well-posedness and positivity for the structured PDE]\label{thm:pde_wellposed_backbone}
Assume $S_0,T_0\in L^1(\mathbb{R}_+)\cap L^\infty(\mathbb{R}_+)$ with $S_0,T_0\ge 0$ a.e., and assume $v_S,v_T\in W^{1,\infty}(\mathbb{R}_+)$ with $v_S,v_T\ge 0$. Suppose the feedback maps $p_1^{\mathrm{tot}},p_2^{\mathrm{tot}},\lambda_P^{\mathrm{tot}},\lambda_R^{\mathrm{tot}}$ are bounded and locally Lipschitz on $[0,\infty)$ with $p_1^{\mathrm{tot}}+p_2^{\mathrm{tot}}\le 1$. Then \eqref{eq:PDE_system_backbone}--\eqref{eq:renewal_bc_backbone} admits a unique mild solution
\[
(S,T)\in C([0,T_{\max});L^1(\mathbb{R}_+)^2),
\]
which is positivity preserving: $S(t,\cdot),T(t,\cdot)\ge 0$ a.e.\ for all $t<T_{\max}$. If \eqref{eq:flux_vanish_backbone} holds, then the totals \eqref{eq:totals_def_backbone} are absolutely continuous and satisfy the exact mass-balance ODE system \eqref{eq:totals_ode_backbone} below.
\end{theorem}

\medskip
\noindent\textbf{Exact reduction to totals.}
Integrating \eqref{eq:PDE_system_backbone} over $x\in(0,\infty)$, using the vanishing flux at infinity \eqref{eq:vanishing_infty} and the renewal boundary conditions \eqref{eq:renewal_bc_backbone}, yields the closed two-dimensional system for the totals:
\begin{equation}\label{eq:totals_ode_backbone}
\begin{cases}
P'(t)=\big(p_1^{\mathrm{tot}}(W(t))-p_2^{\mathrm{tot}}(W(t))\big)\lambda_P^{\mathrm{tot}}(W(t))\,P(t)
+\lambda_R^{\mathrm{tot}}(P(t))\,W(t),
\\[6pt]
W'(t)=\big(1-p_1^{\mathrm{tot}}(W(t))+p_2^{\mathrm{tot}}(W(t))\big)\lambda_P^{\mathrm{tot}}(W(t))\,P(t)
-\big(\delta_0+\lambda_R^{\mathrm{tot}}(P(t))\big)\,W(t).
\end{cases}
\end{equation}
Summing the two equations gives the total population balance
\begin{equation}\label{eq:total_balance_backbone}
\frac{d}{dt}\big(P(t)+W(t)\big)=\lambda_P^{\mathrm{tot}}(W(t))\,P(t)-\delta_0\,W(t),
\end{equation}
reflecting that each division event produces one net cell at the population level, whereas TD death removes one cell.

\subsection{Equilibrium identities: equalization and ratio laws}
\label{subsec:equilibrium_identities}

We record the algebraic identities satisfied by any \emph{positive} equilibrium of the totals ODE \eqref{eq:totals_ode_backbone}. These identities are purely structural: they do not require global stability assumptions and therefore provide robust, interpretable constraints that can be compared across deterministic and stochastic regimes.

Let $(P^*,W^*)\in(0,\infty)^2$ be an equilibrium of \eqref{eq:totals_ode_backbone}, i.e.,
\begin{equation}\label{eq:ss_identities}
\begin{cases}
0=\big(p_1^{\mathrm{tot}}(W^*)-p_2^{\mathrm{tot}}(W^*)\big)\lambda_P^{\mathrm{tot}}(W^*)\,P^*
+\lambda_R^{\mathrm{tot}}(P^*)\,W^*,\\[6pt]
0=\big(1-p_1^{\mathrm{tot}}(W^*)+p_2^{\mathrm{tot}}(W^*)\big)\lambda_P^{\mathrm{tot}}(W^*)\,P^*
-\big(\delta_0+\lambda_R^{\mathrm{tot}}(P^*)\big)\,W^*.
\end{cases}
\end{equation}
Writing \eqref{eq:ss_identities} as a homogeneous linear system in $(P^*,W^*)$ with coefficients evaluated at $(P^*,W^*)$, we obtain
\begin{equation}\label{eq:matrix_form_identities}
\mathbf{M}(P^*,W^*)\binom{P^*}{W^*}=0,\qquad
\mathbf{M}(P,W)\coloneqq
\begin{pmatrix}
(p_1^{\mathrm{tot}}(W)-p_2^{\mathrm{tot}}(W))\lambda_P^{\mathrm{tot}}(W) & \lambda_R^{\mathrm{tot}}(P)\\[4pt]
(1-p_1^{\mathrm{tot}}(W)+p_2^{\mathrm{tot}}(W))\lambda_P^{\mathrm{tot}}(W) & -(\delta_0+\lambda_R^{\mathrm{tot}}(P))
\end{pmatrix}.
\end{equation}
Since $(P^*,W^*)\neq (0,0)$, the coefficient matrix must be singular. Assuming $\lambda_P^{\mathrm{tot}}(W^*)>0$ (as in Section~\ref{subsec:pde_to_ode}), we compute
\begin{align}
0
=\det\mathbf{M}(P^*,W^*)
&=-\lambda_P^{\mathrm{tot}}(W^*)\Big((p_1^{\mathrm{tot}}(W^*)-p_2^{\mathrm{tot}}(W^*))\delta_0+\lambda_R^{\mathrm{tot}}(P^*)\Big),
\label{eq:det_zero_identities}
\end{align}
which yields the \emph{equalization law}
\begin{equation}\label{eq:equalization_law}
\lambda_R^{\mathrm{tot}}(P^*)=\big(p_2^{\mathrm{tot}}(W^*)-p_1^{\mathrm{tot}}(W^*)\big)\delta_0.
\end{equation}
In particular, since $\delta_0>0$ and $\lambda_R^{\mathrm{tot}}\ge 0$, \eqref{eq:equalization_law} implies the necessary sign constraint
\begin{equation}\label{eq:sign_constraint_identities}
p_2^{\mathrm{tot}}(W^*)\ge p_1^{\mathrm{tot}}(W^*).
\end{equation}

Substituting \eqref{eq:equalization_law} into either equation of \eqref{eq:ss_identities} yields the \emph{ratio law}. Indeed, substituting into the first equation gives
\[
\big(p_1^{\mathrm{tot}}(W^*)-p_2^{\mathrm{tot}}(W^*)\big)\Big(\lambda_P^{\mathrm{tot}}(W^*)P^*-\delta_0W^*\Big)=0,
\]
and substituting into the second equation gives
\[
\big(1-p_1^{\mathrm{tot}}(W^*)+p_2^{\mathrm{tot}}(W^*)\big)\Big(\lambda_P^{\mathrm{tot}}(W^*)P^*-\delta_0W^*\Big)=0.
\]
Since $1-p_1^{\mathrm{tot}}(W^*)+p_2^{\mathrm{tot}}(W^*)\ge 0$ and cannot vanish simultaneously with $p_1^{\mathrm{tot}}(W^*)-p_2^{\mathrm{tot}}(W^*)$, we conclude that
\begin{equation}\label{eq:ratio_law}
\frac{P^*}{W^*}=\frac{\delta_0}{\lambda_P^{\mathrm{tot}}(W^*)}.
\end{equation}

We summarize \eqref{eq:equalization_law} and \eqref{eq:ratio_law} as follows.

\begin{lemma}[Equalization and ratio identities]\label{lem:eq_ratio_laws}
Assume $\delta_0>0$ and $\lambda_P^{\mathrm{tot}}(W)>0$ for all $W\ge 0$. A point $(P^*,W^*)\in(0,\infty)^2$ is an equilibrium of \eqref{eq:totals_ode_backbone} if and only if it satisfies
\begin{equation}\label{eq:two_laws}
\lambda_R^{\mathrm{tot}}(P^*)=\big(p_2^{\mathrm{tot}}(W^*)-p_1^{\mathrm{tot}}(W^*)\big)\delta_0,
\qquad
\frac{P^*}{W^*}=\frac{\delta_0}{\lambda_P^{\mathrm{tot}}(W^*)}.
\end{equation}
In particular, \eqref{eq:sign_constraint_identities} is necessary for the existence of a positive equilibrium when $\lambda_R^{\mathrm{tot}}\ge 0$.
\end{lemma}

\begin{proof}
Necessity was established above by the singularity argument. For sufficiency, assume \eqref{eq:two_laws} holds. Substituting these identities into \eqref{eq:totals_ode_backbone} shows that both steady-state equations in \eqref{eq:ss_identities} are satisfied, hence $(P^*,W^*)$ is an equilibrium of \eqref{eq:totals_ode_backbone}.
\end{proof}

\begin{remark}[Degenerate cases]\label{rmk:degenerate_backbone}
If $p_1^{\mathrm{tot}}(W^*)=p_2^{\mathrm{tot}}(W^*)$, then \eqref{eq:equalization_law} forces $\lambda_R^{\mathrm{tot}}(P^*)=0$. Thus balanced symmetric outcomes at equilibrium are compatible only with a vanishing dedifferentiation flux at that equilibrium. Conversely, if $p_2^{\mathrm{tot}}(W)<p_1^{\mathrm{tot}}(W)$ for all $W\ge 0$, then \eqref{eq:equalization_law} cannot hold with $\lambda_R^{\mathrm{tot}}(P^*)\ge 0$, and no positive equilibrium exists.
\end{remark}

\begin{remark}[Interpretation in relation to the stochastic layer]\label{rmk:backbone_interpretation}
The ratio law \eqref{eq:ratio_law} expresses the steady-state compartment ratio in terms of a death--division quotient, while the equalization law \eqref{eq:equalization_law} identifies the dedifferentiation flux required to offset net differentiation bias. Although the stochastic layer (Sections~\ref{sec:ctmc}--\ref{sec:noR}) uses concentration feedback and does not assume this total-based structure, these identities provide a useful interpretive lens. In particular, removing dedifferentiation eliminates the balancing mechanism encoded by \eqref{eq:equalization_law}, consistent with the no-dedifferentiation pathology of Section~\ref{sec:noR}.
\end{remark}

\section{Deterministic dynamics: theorem regime versus threshold phenomenology}
\label{sec:det_dynamics}

The stochastic results in Sections~\ref{sec:ctmc}--\ref{sec:noR} do not require global stability of the deterministic backbone. Nevertheless, deterministic phase-space structure is useful for interpretation and for organizing numerical experiments. In particular, the total ODE \eqref{eq:totals_ode_backbone} can exhibit qualitatively distinct behaviors depending on the feedback architecture. A key point for clarity (and for consistency with the numerical section) is that our deterministic analysis separates two regimes:
\begin{itemize}
\item \emph{Regime A} (the theorem regime), where verifiable monotonicity and sensitivity conditions provide \emph{sufficient} criteria for uniqueness and global convergence to a positive equilibrium;
\item \emph{Regime B} (outside the theorem regime), where bistability, separatrices, and Allee-type thresholds can occur \citep{hart_reconstructing_2014}.
\end{itemize}
This separation resolves an otherwise common source of confusion: numerical examples exhibiting saddle equilibria and threshold dynamics do not contradict Regime~A results, because they arise precisely when the sufficient hypotheses of the theorem are violated. 

Throughout this section we work with the totals ODE \eqref{eq:totals_ode_backbone} and total-based feedback maps
$p_i^{\mathrm{tot}},\lambda_P^{\mathrm{tot}},\lambda_R^{\mathrm{tot}}$ introduced in Section~\ref{sec:backbone}.

\subsection{Regime A: global convergence under monotone feedback (theorem regime)}
\label{subsec:regimeA}

We impose the following standing hypotheses. They are used to ensure global convergence in the theorem regime.

\begin{assumption}[Standing assumptions for the deterministic ODE layer]\label{ass:parameter}
Let $\delta_0>0$ be fixed. Assume that
\[
p_1^{\mathrm{tot}},p_2^{\mathrm{tot}}:[0,\infty)\to[0,1],\qquad \lambda_P^{\mathrm{tot}}:[0,\infty)\to(0,\infty),\qquad \lambda_R^{\mathrm{tot}}:[0,\infty)\to[0,\infty)
\]
are continuously differentiable and locally Lipschitz on $[0,\infty)$. Define $p_3(W)\coloneqq 1-p_1(W)-p_2(W)$.

\begin{enumerate}[label=\textnormal{(A\arabic*)}]
\item \textnormal{(Division outcome probabilities and admissibility).}
The maps $p_1^{\mathrm{tot}},p_2^{\mathrm{tot}}:[0,\infty) \to [0,1]$ satisfy 
\begin{equation}\label{eq:D1_simplex}
p_1^{\mathrm{tot}}(W)+p_2^{\mathrm{tot}}(W)\le 1,\qquad \forall\, W\ge 0.
\end{equation}
Moreover, 
\[
\frac{dp_1^{\mathrm{tot}}}{dW}<0,\qquad \frac{dp_2^{\mathrm{tot}}}{dW}<0,\qquad \forall\,W\ge 0.
\]
Moreover, we assume $p_1^{\mathrm{tot}}(0)>0$ and $p_2^{\mathrm{tot}}(0)>0$.

\item \textnormal{(Feedback rates and boundedness).}
The division rate $\lambda_P^{\mathrm{tot}}$ is bounded and strictly positive on $[0,\infty)$:
\begin{equation}\label{eq:D2_lambdaP}
0<\lambda_{P,\min}\le \lambda_P^{\mathrm{tot}}(W)\le \lambda_{P,\max}<\infty,\qquad \forall\,W\ge 0.
\end{equation}

The map $\lambda_R$ is bounded on $[0,\infty)$:
\begin{equation}\label{eq:D3_lambdaR}
0\le \lambda_R^{\mathrm{tot}}(P)\le \lambda_{R,\max}<\infty,\qquad \forall\,P\ge 0.
\end{equation}

\item \textnormal{Net differentiation bias and compensatory sensitivity.}
There is a persistent net differentiation bias at homeostatic TD levels and a compensatory sensitivity ordering:
\[
p_2^{\mathrm{tot}}(W)>p_1^{\mathrm{tot}}(W),\qquad 
\frac{dp_2^{\mathrm{tot}}}{dW}<\frac{dp_1^{\mathrm{tot}}}{dW}<0,\qquad \forall\,W\ge 0.
\]
Equivalently, differentiation is (in magnitude) more sensitive to TD feedback than self-renewal, while differentiation bias $p_2^{\mathrm{tot}}-p_1^{\mathrm{tot}}$ remains positive.
\end{enumerate}
\end{assumption}
The assumptions above translate established biological regulatory mechanisms into mathematical constraints on the feedback functions.

Assumption (A1) encodes the regulation of cell fate by tissue crowding. Biological evidence indicates that differentiated cells secrete negative feedback regulators that inhibit stem cell proliferation \citep{wu_autoregulation_2003,daluiski_bone_2001,mcpherron_regulation_1997,tzeng_loss_2011,yamasaki_keratinocyte_2003}.
We model this by requiring that an increasing TD load $W$ suppresses both symmetric outcomes ($\dfrac{dp_1^{\mathrm{tot}}}{dW},\dfrac{dp_2^{\mathrm{tot}}}{dW}<0$). Consequently, the system shifts toward asymmetric division ($\dfrac{dp_3^{\mathrm{tot}}}{dW}>0$) as the tissue becomes crowded. 
This reflects a homeostatic strategy: asymmetric division acts as a robust, default mode that maintains the stem pool size without expanding the stem population or differentiated population, a behavior observed in homeostatic stem cell niches \citep{morrison_asymmetric_2006}. 

Assumption (A2) ensures the system is physically realistic and well-posed. The division rate $\lambda_P(W)$ remains positive to allow recovery from injury. Similarly, $\lambda_R^{\mathrm{tot}}(P)$ is bounded, consistent with the limited plasticity of TD cells.

Assumption (A3) is crucial for the stability of the lineage and imposes a two-fold structure on the population dynamics:
\begin{enumerate}[label=(\roman*)]
    \item \textnormal{Net differentiation bias ($p_2^{\mathrm{tot}}>p_1^{\mathrm{tot}}$):} Under normal conditions, the probability of differentiation exceeds that of self-renewal. This creates a new loss on the stem cell pool, ensuring that the lineage does not expand uncontrollably. This loss is balanced at equilibrium by the dedifferentiation flux $\lambda_R$. 
    This condition is in tune with the result in \citep{lo_feedback_2009,doumic_structured_2011} that the proliferation probability of stem cells being smaller than their differentiation probability maintains tissue homeostasis. 
    \item \textnormal{Compensatory sensitivity ($\bigg|\dfrac{dp_2^{\mathrm{tot}}}{dW}\bigg|>\bigg|\dfrac{dp_1^{\mathrm{tot}}}{dW}\bigg|$):} The condition $\dfrac{dp_2^{\mathrm{tot}}}{dW} < \dfrac{dp_1^{\mathrm{tot}}}{dW} < 0$ implies that differentiation is more sensitive to feedback inhibition than self-renewal. As TD load $W$ increases, differentiation shuts down more rapidly than self-renewal. This differential sensitivity preserves the stem cell pool under high regulatory loads, preventing lineage extinction while limiting TD accumulation.
\end{enumerate}

We do not impose a monotonicity constraint for $\lambda_R$. A naturally monotonicity condition for dedifferentiation is 
\begin{equation}\label{eq:lambdaRdecreasing}
\frac{d\lambda_R^{\mathrm{tot}}}{dP}<0,\qquad \forall\,P\ge 0.
\end{equation}
This is biologically motivated by the mechanisms that the direct contact of TD cells with stem cells will suppress their dedifferentiation. It is natural to treat dedifferentiation as a regenerative feedback triggered by stem depletion, so dedifferentiation is stronger when the stem pool is small and is suppressed under stem abundance.
Experimental results even show that the dedifferentiation of mouse secretory cells is prevented by a single stem cell \citep{tata_dedifferentiation_2013}. This suppression can be relieved only when the stem pool is completely ablated \citep{guo_dedifferentiation_2022}. 
However, contrary to the above suppression effect, stem cells have also been documented to synergistically promote dedifferentiation \citep{jiang_involvement_2015}. Cells in proximity to procambium cells, a kind of stem cell-like cells, can dedifferentiate more easily than other cell types during callus formation in Arabidopsis petioles. 

We state a sufficient condition ensuring that the positive equilibrium of \eqref{eq:totals_ode_backbone}, when it exists, is unique and attracts all bounded trajectories in $(0,\infty)^2$. The proof uses a Bendixson--Dulac argument to exclude periodic and homoclinic orbits in the positive orthant.

\begin{theorem}[Uniqueness and global convergence in the theorem regime]\label{thm:global_det_regimeA}
Assume that \eqref{eq:totals_ode_backbone} admits at least one equilibrium $(P^*,W^*)\in(0,\infty)^2$. Suppose that, for all $P,W>0$,
\begin{equation}\label{eq:regimeA_signs}
\big(\lambda_P^{\mathrm{tot}}\big)'(W)<0,
\qquad
0<\big(\lambda_R^{\mathrm{tot}}\big)'(P)\le \frac{\lambda_R^{\mathrm{tot}}(P)}{P}.
\end{equation}
Then the equilibrium in $(0,\infty)^2$ is unique. Moreover, every bounded solution of \eqref{eq:totals_ode_backbone} with initial condition $(P(0),W(0))\in (0,\infty)^2$ satisfies
\[
(P(t),W(t))\longrightarrow (P^*,W^*)
\qquad \text{as } t\to\infty.
\]
\end{theorem}

\begin{proof}[Proof sketch]
The equilibrium identities in Section~\ref{subsec:equilibrium_identities} imply that any positive equilibrium $(P^*,W^*)$ must satisfy the equalization and ratio laws. The ratio law \eqref{eq:ratio_law} allows us to express $P^*$ as a strictly increasing function of $W^*$:
\begin{equation}\label{eq:P_reexpress}
    P^* = \frac{\delta_0W^*}{\lambda_P^{\mathrm{tot}}(W)}.
\end{equation}
Plugging \eqref{eq:P_reexpress} into the equalization law \eqref{eq:equalization_law} yields the following relation:
\begin{equation}\label{eq:two-hand_side}
    (p_2^{\mathrm{tot}}(W^*)-p_1^{\mathrm{tot}}(W^*))\delta_0 = \lambda_R^{\mathrm{tot}}(P^*(W^*)).
\end{equation}
The two hand sides of \eqref{eq:two-hand_side} have the reverse monotonicity, ensuring a unique $W^*$. The monotonicity of $P^*$ in $W^*$ ensures a unique $P^*$, hence the uniqueness of the equilibrium $(P^*,W^*)$.

To preclude closed orbits in the planar system, choose the Dulac function $\phi(P,W)=1/(PW)$ on $(0,\infty)^2$ and compute
\[
\frac{\partial(\phi G_1)}{\partial P}+\frac{\partial(\phi G_2)}{\partial W},
\]
where $(G_1,G_2)$ is the vector field of \eqref{eq:totals_ode_backbone}. Under Assumption~\ref{ass:parameter} and \Cref{eq:regimeA_signs}, this divergence is strictly negative on $(0,\infty)^2$, so Dulac's criterion excludes periodic orbits. A standard Green's theorem argument then also excludes homoclinic loops. By the Poincar\'e--Bendixson theorem \citep{perko_differential_2001}, every bounded trajectory converges to an equilibrium; uniqueness implies convergence to $(P^*,W^*)$.
\end{proof}

\begin{remark}[Role in the paper]\label{rmk:regimeA_role}
Theorem~\ref{thm:global_det_regimeA} provides \emph{sufficient} and directly checkable conditions guaranteeing a globally attracting homeostatic equilibrium in the deterministic totals model. In Section~\ref{sec:numerics} we use a parameter set satisfying \eqref{eq:regimeA_signs} to illustrate global convergence (Numerical set~III-A), thereby separating theorem-consistent dynamics from the threshold phenomena discussed next.
\end{remark}

\subsection{Regime B: bistability, separatrices, and Allee-type thresholds outside the theorem regime}
\label{subsec:regimeB}

When the monotonicity and sensitivity conditions \eqref{eq:regimeA_signs} are violated, the planar totals system \eqref{eq:totals_ode_backbone} may exhibit richer phase portraits, including saddle equilibria, invariant manifolds, and bistability between extinction and a positive homeostatic state. Such threshold phenomena are biologically meaningful: they encode a critical-mass requirement and predict failure of recovery when the lineage is depleted beyond a separatrix. Furthermore, they are not specific to lineage models: in related population systems, algebraic--spectral conditions can organize existence windows of interior equilibria and delineate saddle--node thresholds, providing a reproducible bifurcation-level description of bistability and regime transitions \citep{wang_algebraicspectral_2026}.

In particular, our numerical examples (Section~\ref{sec:numerics}) exhibit a positive equilibrium with one stable and one unstable eigen-direction (a hyperbolic saddle) together with a separatrix that partitions the positive orthant into distinct basins of attraction. Trajectories on one side of the separatrix converge to the positive equilibrium, while trajectories on the other side collapse to extinction, producing an Allee-type effect in the deterministic dynamics.

\begin{remark}[Consistency with the theorem regime]\label{rmk:regimeB_consistency}
These bistable and Allee-type behaviors occur for parameter sets that violate the sufficient hypotheses of Theorem~\ref{thm:global_det_regimeA} (for instance, when $\big(\lambda_R^{\mathrm{tot}}\big)'(P)\le \lambda_R^{\mathrm{tot}}(P)/P$ fails, or when alternative sensitivity orderings are imposed). They therefore do \emph{not} contradict the conclusions of Regime~A. Instead, they delineate the boundary of the theorem regime and quantify the threshold risks that emerge when those sufficient stability conditions are not met, including collapse after strong depletion and recovery only above a critical mass.
\end{remark}

\begin{remark}[How Regime B interfaces with the stochastic layer]\label{rmk:regimeB_stochastic}
In the stochastic setting, demographic noise can induce rare transitions across deterministic separatrices and can amplify extinction risk in the low-density basin. Conversely, in the presence of dedifferentiation, stochastic trajectories initialized well inside the homeostatic basin may remain concentrated near the deterministic attractor for large system size. We return to these points in Section~\ref{sec:discussion}.
\end{remark}

\section{Numerical experiments}
\label{sec:numerics}

This section provides numerical evidence supporting the analytical results developed in Sections~\ref{sec:ctmc}--\ref{sec:det_dynamics}. The experiments are organized by theorem appearance order to avoid conflating deterministic stability claims with stochastic scaling statements. In particular:
\begin{enumerate}[label=(\roman*)]
\item Section~\ref{subsec:num_diffusion} validates the diffusion approximation and the $\Omega^{-1/2}$ fluctuation scaling predicted by the FCLT (\Cref{sec:diffusion,sec:stochastic_analysis}).
\item Section~\ref{subsec:num_noR} illustrates the no-dedifferentiation pathology (Section~\ref{sec:noR}) in subcritical and supercritical regimes.
\item Section~\ref{subsec:num_regimes} contrasts deterministic phase portraits in Regime~A (theorem regime) and Regime~B (threshold phenomenology), supporting the regime separation in Section~\ref{sec:det_dynamics}.
\end{enumerate}

All simulations were performed in Python (SciPy/NumPy). Stochastic trajectories of the diffusion model were generated by the Euler--Maruyama method with sufficiently small time steps; parameter values and initial conditions are stated in each subsection.

\subsection{Validation of the diffusion approximation and \texorpdfstring{$\Omega^{-1/2}$} scaling}
\label{subsec:num_diffusion}

We first validate the diffusion approximation derived from the density-dependent CTMC via the FCLT (Theorem~\ref{thm:fclt}). The objective here is not to assess global stability of the deterministic model, but rather to demonstrate two scaling features of the diffusion approximation:
\begin{enumerate}[label=\textnormal{(\roman*)}]
\item stochastic trajectories fluctuate around the mean-field solution, and
\item the fluctuation amplitude decreases as $\mathcal{O}(\Omega^{-1/2})$ as the system size increases.
\end{enumerate}

We use the same functional forms as in the text for the concentration-feedback maps:
\[
p_1(w)=\frac{0.3}{1+w/150},\qquad
p_2(w)=\frac{0.35}{1+w/150}+0.1,\qquad
\lambda_P(w)=\frac{2}{1+(w/20)^3},\qquad
\lambda_R(p)=\frac{0.1\,p}{50+p},
\]
with TD death rate $\delta_0=0.5$. We simulate the chemical Langevin diffusion \eqref{eq:cle_main} for two system sizes, $\Omega=50$ (higher noise) and $\Omega=1000$ (lower noise), and compare trajectories to the deterministic mean-field ODE \eqref{eq:mean_field_ode}.

\Cref{fig:mean_field} shows trajectories initiated from a low-density initial condition. The stochastic paths fluctuate around the deterministic trajectory, with visibly larger deviations at $\Omega=50$ than at $\Omega=1000$, consistent with the $\Omega^{-1/2}$ scaling. Importantly, this experiment is intended to validate the diffusion approximation and its scaling properties, not to establish deterministic stability properties of the underlying ODE.

\begin{figure}[htbp]
    \centering
    \includegraphics[width=119mm]{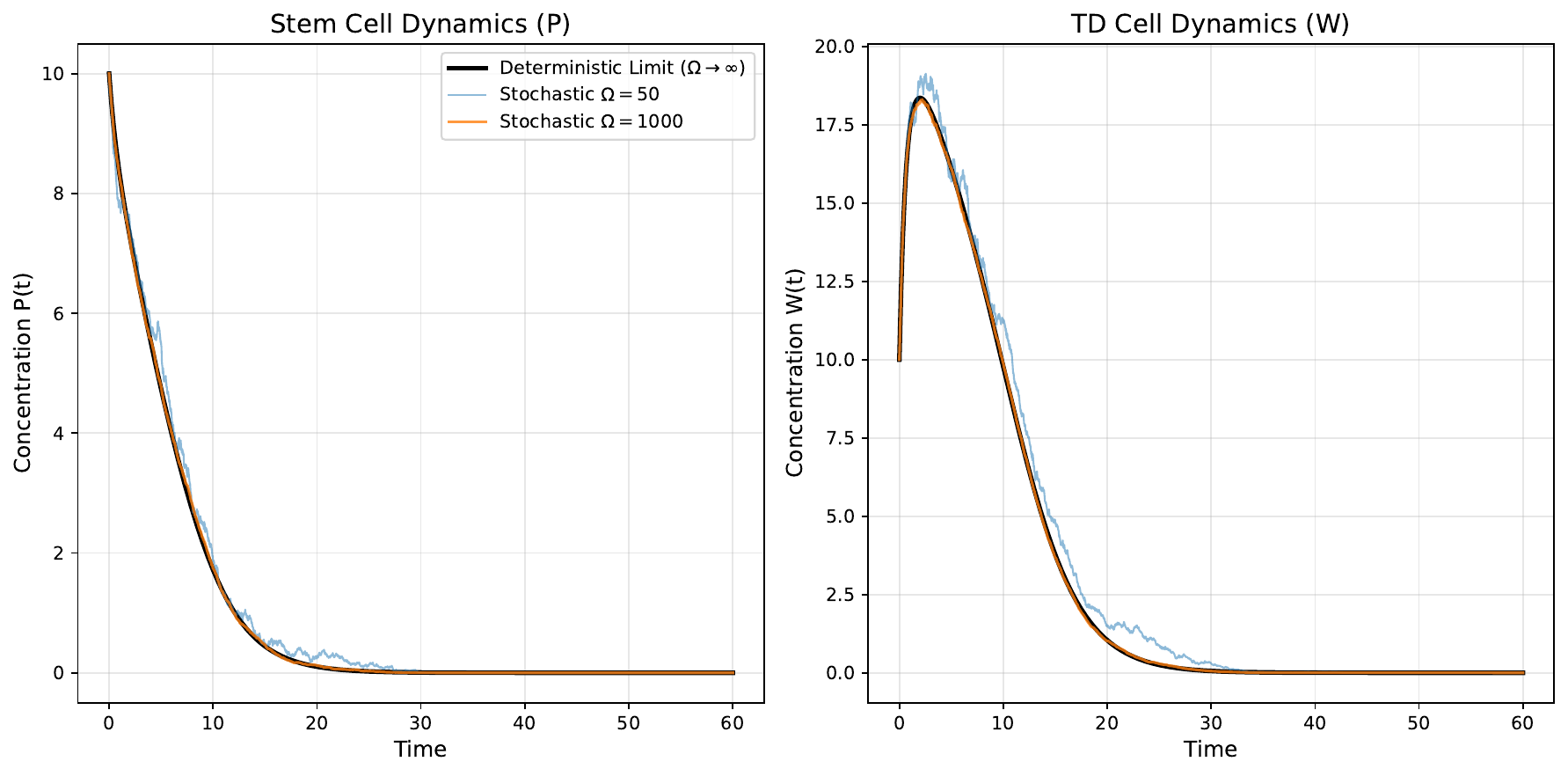}
    \caption{Diffusion approximation and system-size scaling. The black curve shows the deterministic mean-field trajectory. Colored curves show sample paths of the chemical Langevin diffusion \eqref{eq:cle_main} for $\Omega=50$ (higher noise) and $\Omega=1000$ (lower noise). The reduction in fluctuation amplitude with increasing $\Omega$ is consistent with the $\mathcal{O}(\Omega^{-1/2})$ scaling predicted by Theorem~\ref{thm:fclt}.}
    \label{fig:mean_field}
\end{figure}

To complement the low-density initialization, \Cref{fig:diffusion_approximation} shows simulations started closer to the homeostatic scale. Again, stochastic realizations remain clustered around the deterministic trajectory for $\Omega=1000$ and exhibit larger excursions for $\Omega=50$. Together, \Cref{fig:mean_field,fig:diffusion_approximation} corroborate that the diffusion approximation captures mean-field drift to leading order while quantifying demographic fluctuations through the mechanism-consistent covariance structure.

\begin{figure}[htbp]
    \centering
    \includegraphics[width=119mm]{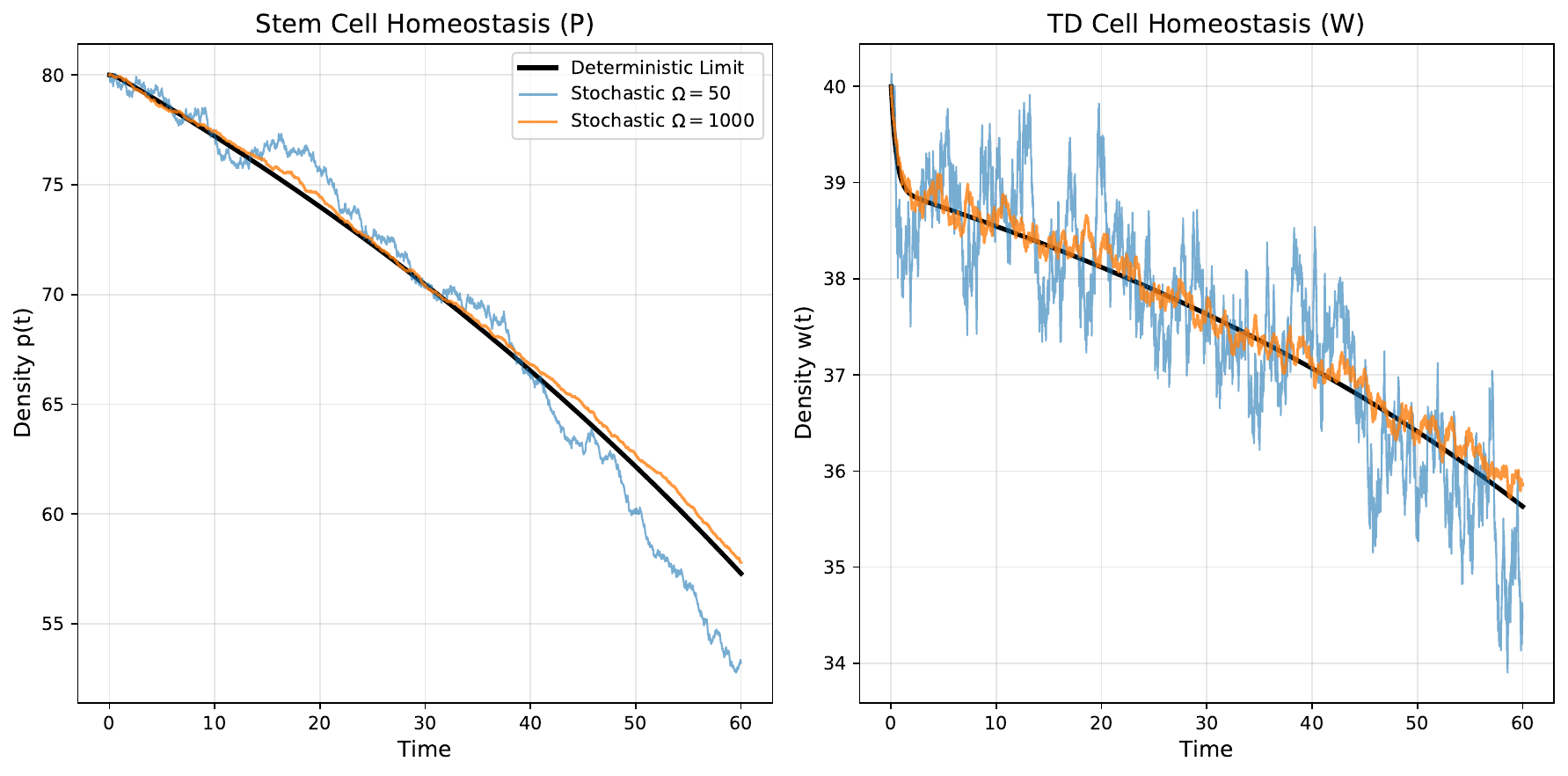}
    \caption{\textbf{Validation of diffusion approximation in a higher-density regime.} The black curve shows the deterministic mean-field trajectory, while colored curves show chemical Langevin sample paths for $\Omega=50$ and $\Omega=1000$. The fluctuation magnitudes of trajectories decrease with $\Omega$, consistent with $\mathcal{O}(\Omega^{-1/2})$ fluctuations.}
    \label{fig:diffusion_approximation}
\end{figure}

\subsection{Numerical illustration of the no-dedifferentiation pathology}
\label{subsec:num_noR}

We next illustrate the dichotomy of Proposition~\ref{prop:no_dedifferentiation_main} by simulating the diffusion approximation with dedifferentiation removed ($\lambda_R\equiv 0$). The simulations are partitioned into subcritical and supercritical regimes according to the sign of the net stem drift $p_1-p_2$.

\Cref{fig:pathology} (left) shows a subcritical regime in which the stem coordinate exhibits asymptotic extinction in the sense of Definition~\ref{def:asymptotic_extinction}. \Cref{fig:pathology} (right) shows a supercritical regime in which moments grow rapidly, consistent with the exponential moment divergence statement in Proposition~\ref{prop:no_dedifferentiation_main}. We emphasize that these figures are intended to visualize long-time decay versus growth in the diffusion approximation; they should not be interpreted as proofs of finite-time boundary hitting properties, which depend on boundary attainability and numerical boundary handling.

\begin{figure}[htbp]
    \centering
    \includegraphics[width=119mm]{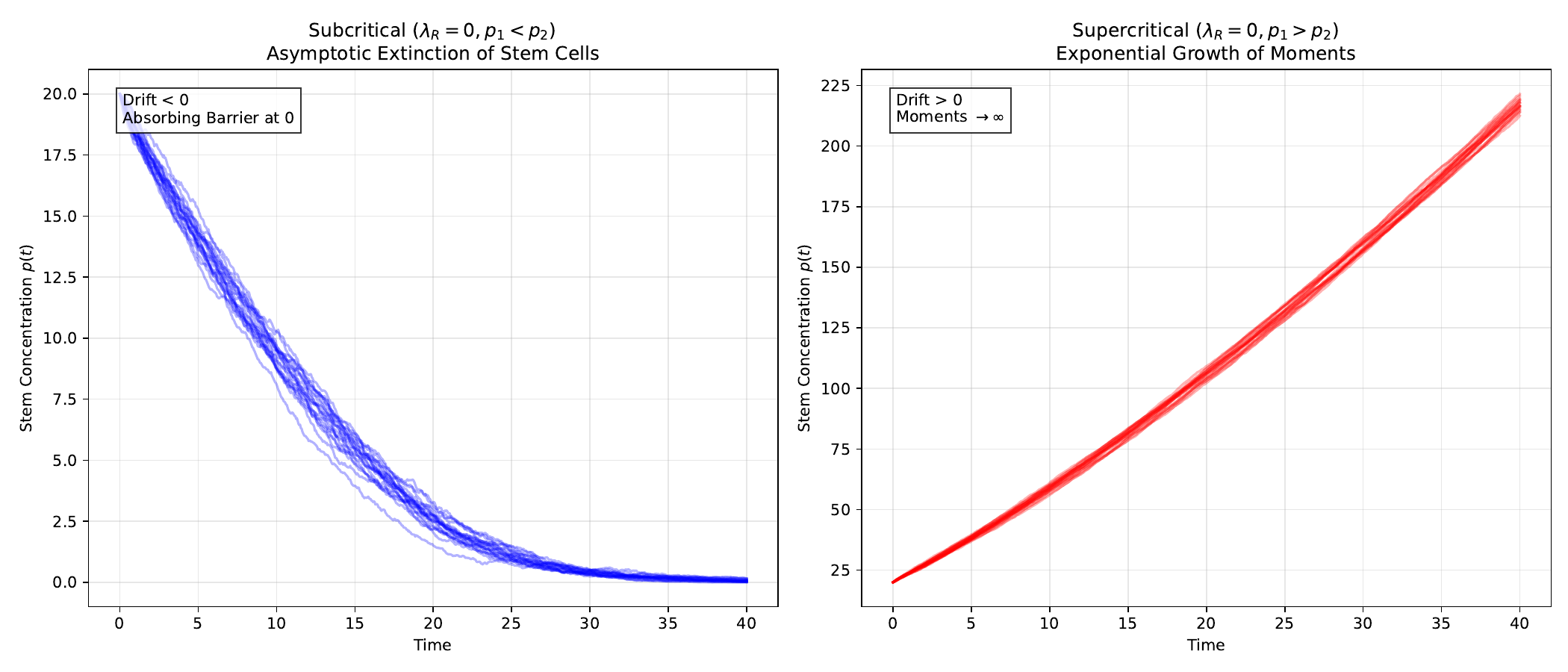}
    \caption{\textbf{No-dedifferentiation pathology in the diffusion approximation ($\lambda_R\equiv 0$).} (a): subcritical regime (net stem drift uniformly negative) exhibiting asymptotic extinction of the stem coordinate. (b): supercritical regime (net stem drift uniformly positive) exhibiting rapid growth consistent with exponential moment divergence. These simulations illustrate Proposition~\ref{prop:no_dedifferentiation_main} at the diffusion-approximation level and do not assert finite-time boundary hitting.}
    \label{fig:pathology}
\end{figure}

\subsection{Regime A versus Regime B deterministic phase portraits}
\label{subsec:num_regimes}

This subsection complements the stochastic scaling tests by illustrating the deterministic regime separation introduced in Section~\ref{sec:det_dynamics}. The goal is to (i) provide a representative \emph{Regime~A} example satisfying the ``theorem-regime'' hypotheses of Theorem~\ref{thm:global_det_regimeA} (global convergence to a unique positive equilibrium for bounded trajectories), and (ii) provide a contrasting \emph{Regime~B} example exhibiting a saddle equilibrium, a separatrix, and an Allee-type threshold. Importantly, the Regime~B example is constructed to lie \emph{outside} the sufficient hypotheses of Theorem~\ref{thm:global_det_regimeA}, thereby avoiding any logical tension between deterministic theory and numerical phenomenology.

\paragraph{Regime A (theorem regime: global convergence)}
For Regime~A we use the parameter set reported in Section~\ref{subsec:num_diffusion} (Deterministic set~III-A), for which all tested trajectories converge to a unique positive equilibrium. In particular, the numerical experiments yield convergence to
\[
(P^*,W^*) \approx (50.9895,\ 54.0855),
\]
with final distances to $(P^*,W^*)$ at $t=400$ on the order of $10^{-7}$ across a wide grid of initial conditions. This provides a representative visualization of theorem-consistent global convergence and will be used as the Regime~A reference point in the discussion.

\paragraph{Regime B (threshold phenomenology outside the theorem regime)}
To illustrate threshold dynamics, we consider a total-based feedback architecture that produces two positive equilibria: a hyperbolic saddle and a locally asymptotically stable equilibrium. We use the same totals-fate and division-rate maps as in the preceding deterministic examples,
\[
p_1^{\mathrm{tot}}(W)=\frac{0.30}{1+W/150},\qquad
p_2^{\mathrm{tot}}(W)=\frac{0.35}{1+W/150}+0.10,\qquad
\lambda_P^{\mathrm{tot}}(W)=\frac{2}{1+(W/20)^3},\qquad \delta_0=0.5,
\]
and choose a non-monotone dedifferentiation feedback of the form
\begin{equation}\label{eq:regimeB_lambdaR}
\lambda_R^{\mathrm{tot}}(P)
=
r_{\max}\,
\frac{P^{n}}{K^{n}+P^{n}}\,
\exp\!\left(-\frac{P}{M}\right),
\qquad (n,K,M)=(3,60,80),
\end{equation}
with $r_{\max}=0.2712$. (This choice is intended as a constructive example of behavior outside the theorem regime; see Remark~\ref{rmk:regimeB_consistency} for interpretation.)

For this Regime~B parameter set, the steady-state compatibility equation (ratio + equalization laws, Lemma~\ref{lem:eq_ratio_laws}) admits two positive equilibria. Linearization of \eqref{eq:totals_ode_backbone} at these equilibria gives:
\begin{itemize}
\item a \emph{hyperbolic saddle} at
\[
(P_{\mathrm{sd}},W_{\mathrm{sd}})\approx (72.7641,\ 37.7305),
\qquad
\mathrm{spec}(J)\approx\{-2.0982,\ 0.00575\},
\]
\item and a \emph{locally asymptotically stable} equilibrium at
\[
(P_{\mathrm{st}},W_{\mathrm{st}})\approx (82.4801,\ 39.0566),
\qquad
\mathrm{spec}(J)\approx\{-2.1155,\ -0.00516\}.
\]
\end{itemize}
The stable manifold of the saddle forms a separatrix that partitions the positive quadrant into two basins: trajectories initialized on one side converge to the stable positive equilibrium, whereas trajectories initialized on the other side collapse to extinction. This produces an Allee-type threshold in the deterministic total dynamics.

\paragraph{Quantitative evidence that Regime B lies outside the theorem regime}
Theorem~\ref{thm:global_det_regimeA} requires, among other conditions, the inequality
\[
\big(\lambda_R^{\mathrm{tot}}\big)'(P)\le \frac{\lambda_R^{\mathrm{tot}}(P)}{P},
\qquad \forall\,P>0.
\]
For the Regime~B feedback \eqref{eq:regimeB_lambdaR}, we verify numerically that this sufficient condition is violated: there exists $P>0$ such that
\[
\big(\lambda_R^{\mathrm{tot}}\big)'(P) - \frac{\lambda_R^{\mathrm{tot}}(P)}{P} > 0,
\]
and in fact
\[
\max_{P>0}\left(\big(\lambda_R^{\mathrm{tot}}\big)'(P)-\frac{\lambda_R^{\mathrm{tot}}(P)}{P}\right)
\approx 9.02\times 10^{-4}
\quad \text{at } P\approx 3.21\times 10^{1}.
\]
Hence this example lies outside the sufficient hypotheses of Theorem~\ref{thm:global_det_regimeA}; the observed bistability and threshold behavior are therefore consistent with the regime classification and do not contradict the theorem-regime convergence statement.

To visualize the threshold dynamics described above, we plot the Regime~B phase portrait and basin map in \Cref{fig:regimeB_phase,fig:regimeB_basin}. 
\Cref{fig:regimeB_phase} displays streamlines of the totals ODE \eqref{eq:totals_ode_backbone}, the nullclines, and the two positive equilibria (the hyperbolic saddle and the locally stable node); the saddle's stable manifold is highlighted as the separatrix. 
\Cref{fig:regimeB_basin} shows a numerical basin classification of initial conditions together with the same separatrix overlaid, confirming that the stable manifold coincides with the basin boundary and hence represents the deterministic Allee-type threshold.

\begin{figure}[htbp]
    \centering
    \includegraphics[width=119mm]{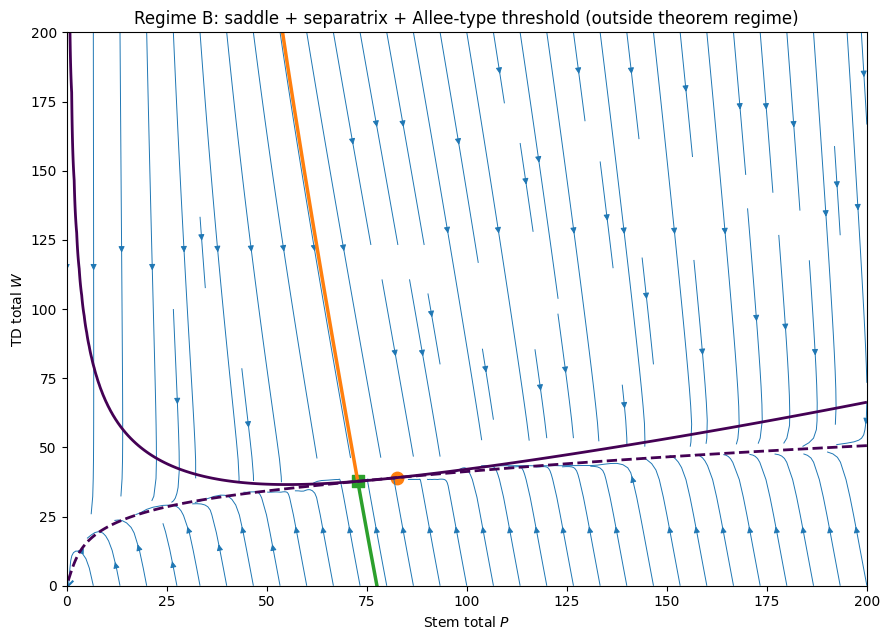}
    \caption{\textbf{Regime B phase portrait: saddle, separatrix, and Allee-type threshold (outside the theorem regime).}
    Streamlines show the vector field of the totals ODE \eqref{eq:totals_ode_backbone} under the Regime~B parameter set \eqref{eq:regimeB_lambdaR}.
    Solid and dashed curves denote the nullclines $\dot P=0$ and $\dot W=0$, respectively.
    The system admits two positive equilibria: a hyperbolic saddle at $(P_{\mathrm{sd}},W_{\mathrm{sd}})\approx(72.76,37.73)$ and a locally stable equilibrium at $(P_{\mathrm{st}},W_{\mathrm{st}})\approx(82.48,39.06)$.
    The stable manifold of the saddle (thick curve) acts as a separatrix dividing trajectories that recover to the positive equilibrium from those that collapse to extinction, producing an Allee-type threshold.}
    \label{fig:regimeB_phase}
\end{figure}

\begin{figure}[htbp]
    \centering
    \includegraphics[width=119mm]{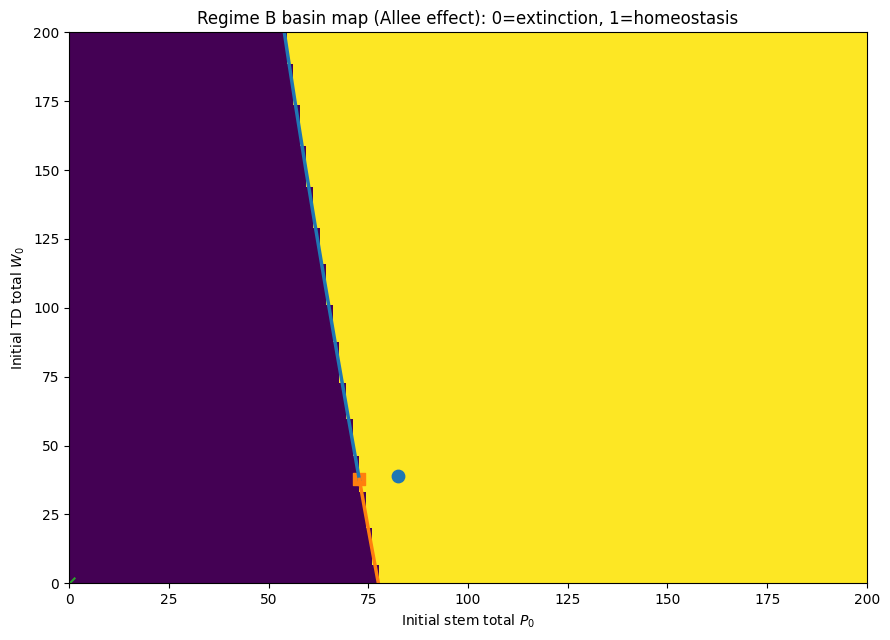}
    \caption{\textbf{Regime B basin map (Allee effect) with separatrix overlay.}
    Each initial condition $(P_0,W_0)$ is classified by its long-time fate under the totals ODE \eqref{eq:totals_ode_backbone} with the Regime~B parameter set \eqref{eq:regimeB_lambdaR}:
    label $0$ (dark region) indicates convergence to extinction, while label $1$ (bright region) indicates convergence to the stable positive equilibrium.
    The overlaid curve is the stable manifold of the saddle equilibrium, computed by backward integration from a small perturbation along the stable eigenvector; it coincides with the basin boundary and therefore represents the deterministic Allee-type threshold.}
    \label{fig:regimeB_basin}
\end{figure}

\section{Discussion and biological implications}
\label{sec:discussion}

This work develops a multiscale, mechanism-consistent framework for stem--TD lineage dynamics with dedifferentiation and feedback. The central theme is that \emph{lineage plasticity is not merely an additional reaction channel}: when demographic fluctuations are derived from an underlying event model, the presence or absence of a dedifferentiation flux has structural consequences for long-time behavior. Our results connect three layers---a density-dependent CTMC, its diffusion approximation, and a deterministic structured backbone---and clarify which conclusions belong to which layer.

\paragraph{Mechanism-consistent noise: from CTMC to diffusion}
We began with a five-channel density-dependent CTMC encoding symmetric self-renewal, symmetric differentiation, asymmetric division, dedifferentiation, and TD death. The FCLT yields a chemical Langevin diffusion with \emph{state-dependent} covariance that exactly matches the aggregated channel-wise infinitesimal covariances. This avoids ad hoc noise prescriptions and ensures that the stochastic model preserves the deterministic mean-field drift while correctly representing the mechanism generating demographic noise. The explicit covariance structure also provides a direct route to quantitative fluctuation predictions and system-size scaling.

\paragraph{No-dedifferentiation pathology as a diffusion-level structural dichotomy}
Our main qualitative conclusion is the no-dedifferentiation pathology (Section~\ref{sec:noR}). At the level of the diffusion approximation, removing the dedifferentiation flux induces a dichotomy: subcritical regimes exhibit almost sure asymptotic extinction of the stem coordinate, while supercritical regimes display exponential divergence of polynomial moments. This is a \emph{diffusion-approximation statement}: it characterizes the behavior of the chemical Langevin model derived under the system-size scaling and does not claim finite-time boundary hitting, nor does it replace the discrete CTMC near-boundary behavior. Nevertheless, it identifies a robust failure mode in the stochastic approximation: without a dedifferentiation mechanism, demographic noise cannot be consistently counteracted by the remaining feedback structure, leading to collapse or runaway growth depending on the sign of the net stem drift.

\paragraph{Deterministic backbone: observable steady-state constraints}
Although the stochastic analysis does not rely on global stability of the deterministic model, the deterministic backbone yields two exact algebraic constraints at positive equilibria: the ratio law and the equalization law (Section~\ref{sec:backbone}). These identities link microscopic fate bias and turnover parameters to macroscopic compartment ratios and to the dedifferentiation flux required to offset net differentiation bias. In applications, they can be used as \emph{observable constraints}: measuring compartment ratios and turnover rates provides direct information about the effective balance between fate bias and plasticity at homeostasis. From an interpretive standpoint, these laws also clarify what fails in the no-dedifferentiation setting: the balancing condition encoded by equalization is structurally absent when $\lambda_R\equiv 0$.

\paragraph{Regime A versus Regime B: stability guarantees and threshold risks}
To address the coexistence of theorem-type convergence claims and threshold dynamics, we distinguished a theorem regime (Regime~A) from a phenomenological regime (Regime~B) in the deterministic totals model (Section~\ref{sec:det_dynamics}). Regime~A provides sufficient and checkable conditions guaranteeing a unique positive equilibrium attracting bounded trajectories. Regime~B, arising when these sufficient conditions are violated, exhibits bistability and separatrices, producing an Allee-type threshold: below a critical initial mass the lineage collapses, while above it recovery to a positive state occurs. Biologically, Regime~B provides a mechanism for \emph{threshold failure} in tissue repair and transplantation: sufficiently depleted systems may be unable to trigger an effective regenerative dedifferentiation flux and therefore enter a collapse basin. This regime classification strengthens the overall narrative: the theorem regime quantifies when global homeostasis is guaranteed, while the phenomenological regime identifies concrete risks when those guarantees fail.

\paragraph{Limitations and outlook}
Several extensions are natural. First, boundary behavior for the diffusion approximation warrants careful comparison with the underlying CTMC in small-population regimes, where discreteness and diffusion degeneracy may lead to boundary-layer effects. 
Second, spatial heterogeneity and microenvironmental feedback can be incorporated, and they can create pattern-forming niches that alter growth and treatment response; hybrid PDE--agent-based models with bidirectional endothelial feedback provide a concrete mechanism by which reaction--diffusion patterning generates heterogeneous drug penetration and resistant niches \citep{liu_bidirectional_2025}
Third, empirical parameterization of feedback maps would enable direct confrontation with lineage-tracing and perturbation experiments, with the ratio/equalization identities serving as calibration constraints and the diffusion model quantifying fluctuation-driven failure probabilities.

Overall, the framework presented here provides a principled bridge from discrete lineage events to noise-driven tissue dynamics and clarifies how a dedifferentiation flux can mitigate demographic fluctuations, while also delineating parameter regimes in which threshold failures can occur.

\backmatter

\section*{Data Availability:}
Data sharing is not applicable to this article as no new data were created or analyzed in this study.

\section*{Funding:}
Both authors of this article have confirmed that this research received no external funding.

\section*{Author Contributions}
All authors contributed to the study. All authors read and approved the final
manuscript.

\section*{Ethical statement}
The authors declare that the research presented in this manuscript is original and has not
been published elsewhere and is not under consideration by another journal. The study was conducted in
accordance with ethical principles and guidelines.

\section*{Declaration:}
All authors declare no competing interests.

\begin{appendices}

\section{Stochastic technicalities}
\label{app:stoch_tech}

This appendix collects technical ingredients used in the stochastic layer:
(i) a uniform LLN for the Poisson time-change representation,
(ii) verification of the Ethier--Kurtz hypotheses for the density-dependent CTMC limits,
and (iii) cutoff/extension tools and auxiliary inequalities used to construct stopped strong solutions
and to streamline moment estimates.

\subsection{A uniform law of large numbers for Poisson processes}
\label{app:poisson_lln}

This appendix provides a self-contained proof of Lemma~\ref{lem:poisson_lln} (uniform LLN for the Poisson time-change), using exponential moment bounds for the compensated Poisson process together with a maximal inequality and the Borel--Cantelli lemma. The argument is standard; we include it for completeness.

\subsubsection{Exponential moment bound for the compensated Poisson process}
\label{app:poisson_mgf}

Let $\{Y(t)\}_{t\ge 0}$ be a unit-rate Poisson process and define the compensated process
\[
\tilde Y(t)\coloneqq Y(t)-t.
\]

\begin{lemma}[Exponential moment bound]\label{lem:poisson_exp_moment}
There exists a constant $C>0$ such that for all $t\ge 0$ and all $\lambda\in[-1,1]$,
\begin{equation}\label{eq:poisson_exp_moment}
\mathbb{E}\Big[e^{\lambda \tilde Y(t)}\Big]\le \exp\big(C\,t\,\lambda^2\big).
\end{equation}
\end{lemma}

\begin{proof}
Since $Y(t)\sim\mathrm{Poi}(t)$, its moment generating function is
\[
\mathbb{E}\big[e^{\lambda Y(t)}\big]=\exp\big(t(e^\lambda-1)\big).
\]
Therefore,
\[
\mathbb{E}\big[e^{\lambda \tilde Y(t)}\big]
=\mathbb{E}\big[e^{\lambda(Y(t)-t)}\big]
=\exp\big(t(e^\lambda-1-\lambda)\big).
\]
For $\lambda\in[-1,1]$, Taylor's theorem yields
\[
e^\lambda-1-\lambda=\frac{e^\xi}{2}\lambda^2
\quad\text{for some }\xi\in[-1,1],
\]
hence $0\le e^\lambda-1-\lambda\le \frac{e}{2}\lambda^2$. Taking $C\ge e/2$ gives \eqref{eq:poisson_exp_moment}.
\end{proof}

\subsubsection{Uniform LLN on compact time intervals}
\label{app:poisson_uniform_lln}

We now prove the uniform LLN used in the mean-field argument.

\begin{lemma}[Uniform LLN]\label{lem:poisson_uniform_lln}
Let $Y$ be a unit-rate Poisson process. For each $u_0>0$,
\[
\lim_{\Omega\to\infty}\ \sup_{0\le u\le u_0}\left|\Omega^{-1}Y(\Omega u)-u\right|=0
\qquad \text{a.s.}
\]
Equivalently,
\[
\lim_{\Omega\to\infty}\ \sup_{0\le u\le u_0}\left|\Omega^{-1}\tilde Y(\Omega u)\right|=0
\qquad \text{a.s.}
\]
\end{lemma}

\begin{proof}
Fix $u_0>0$ and $\varepsilon>0$. Define the event
\[
E_\Omega(\varepsilon)\coloneqq
\left\{\sup_{0\le u\le u_0}\left|\Omega^{-1}\tilde Y(\Omega u)\right|>4\varepsilon\right\},
\qquad
q_\Omega(\varepsilon)\coloneqq \mathbb{P}\big(E_\Omega(\varepsilon)\big).
\]
We show that $\sum_{\Omega=1}^\infty q_\Omega(\varepsilon)<\infty$, which implies
$\mathbb{P}(E_\Omega(\varepsilon)\ \text{i.o.})=0$ by Borel--Cantelli, and hence the claim.

Let $T\coloneqq \Omega u_0$. Since $u\mapsto \tilde Y(u)$ is a c\`adl\`ag martingale with jumps of size $1$,
a standard maximal inequality for sums of independent increments (Etemadi's inequality \citep{etemadi_classical_1985} applied to the increment sequence on a fine partition) yields a constant factor bound of the form
\begin{equation}\label{eq:etemadi_style_bound}
q_\Omega(\varepsilon)
\le 4\,\sup_{0\le u\le u_0}\mathbb{P}\left(\left|\Omega^{-1}\tilde Y(\Omega u)\right|>\varepsilon\right).
\end{equation}
We next bound the right-hand side uniformly over $u\in[0,u_0]$ by an exponential Markov argument.

Fix $u\in[0,u_0]$ and write $t=\Omega u\in[0,T]$. For any $\lambda\in(0,1]$,
\[
\mathbb{P}\big(\tilde Y(t)>\varepsilon\Omega\big)
=\mathbb{P}\big(e^{\lambda \tilde Y(t)} > e^{\lambda \varepsilon\Omega}\big)
\le e^{-\lambda\varepsilon\Omega}\,\mathbb{E}\big[e^{\lambda \tilde Y(t)}\big]
\le \exp\big(-\lambda\varepsilon\Omega + C t\lambda^2\big)
\le \exp\big(-\lambda\varepsilon\Omega + C \Omega u_0 \lambda^2\big),
\]
where we used Lemma~\ref{lem:poisson_exp_moment} and $t\le \Omega u_0$.
Similarly,
\[
\mathbb{P}\big(-\tilde Y(t)>\varepsilon\Omega\big)
=\mathbb{P}\big(e^{-\lambda\tilde Y(t)}>e^{\lambda\varepsilon\Omega}\big)
\le \exp\big(-\lambda\varepsilon\Omega + C \Omega u_0 \lambda^2\big)
\quad \text{for }\lambda\in(0,1].
\]
Choosing $\lambda=\Omega^{-1/2}\in(0,1]$ for $\Omega\ge 1$ gives
\[
\sup_{0\le u\le u_0}\mathbb{P}\left(\left|\Omega^{-1}\tilde Y(\Omega u)\right|>\varepsilon\right)
\le 2\exp\big(-\varepsilon\sqrt{\Omega}+C u_0\big).
\]
Substituting into \eqref{eq:etemadi_style_bound} yields
\[
q_\Omega(\varepsilon)\le 8\exp\big(-\varepsilon\sqrt{\Omega}+C u_0\big),
\]
and therefore $\sum_{\Omega=1}^\infty q_\Omega(\varepsilon)<\infty$.
Hence $E_\Omega(\varepsilon)$ occurs only finitely many times almost surely. Since $\varepsilon>0$ was arbitrary, this implies
\[
\sup_{0\le u\le u_0}\left|\Omega^{-1}\tilde Y(\Omega u)\right|\xrightarrow[\Omega\to\infty]{}0
\quad\text{a.s.},
\]
which is equivalent to the stated uniform LLN for $\Omega^{-1}Y(\Omega u)-u$.
\end{proof}

\begin{remark}
Lemma~\ref{lem:poisson_uniform_lln} is used in Section~\ref{subsec:lln} only through the fact that scaled centered Poisson terms vanish uniformly on compact time intervals. 
\end{remark}

\subsection{Verification of the Ethier--Kurtz hypotheses for the CTMC limits}
\label{app:ek_verification}

In Section~\ref{sec:diffusion} we appeal to standard limit theorems for density-dependent Markov chains (LLN and FCLT) as developed, for example, by Ethier and Kurtz \citep{ethier_markov_1986}. This appendix verifies that the hypotheses required for \Cref{thm:mean_field_limit,thm:fclt} hold for the propensity structure introduced in \Cref{sec:ctmc}.

\subsubsection{Local regularity and growth on compact sets}
\label{app:ek_local}

Recall that the feedback maps $p_1,p_2,\lambda_P,\lambda_R$ are assumed locally Lipschitz on $[0,\infty)$ (\Cref{subsec:ctmc_scaling}), and that the propensity densities are
\[
\begin{aligned}
\alpha_{\mathrm{SR}}(p,w)&=p_1(w)\lambda_P(w)p,\\[4pt]
\alpha_{\mathrm{SD}}(p,w)&=p_2(w)\lambda_P(w)p,\\[4pt]
\alpha_{\mathrm{ASD}}(p,w)&=p_3(w)\lambda_P(w)p,\\[4pt]
\alpha_{\mathrm{R}}(p,w)&=\lambda_R(p)w,\\[4pt]
\alpha_{\mathrm{D}}(p,w)&=\delta_0 w,
\end{aligned}
\]
with $p_3(w)=1-p_1(w)-p_2(w)\ge 0$.

\begin{lemma}[Local Lipschitz continuity of propensities]\label{lem:alpha_loc_lip}
Each propensity density $\alpha_e:\mathbb{R}_+^2\to[0,\infty)$ is locally Lipschitz on $\mathbb{R}_+^2$.
\end{lemma}

\begin{proof}
Fix a compact set $K\subset\mathbb{R}_+^2$. On $K$, the coordinates $p,w$ are bounded, and the one-dimensional maps
$p_1,p_2,p_3,\lambda_P,\lambda_R$ are Lipschitz on the corresponding bounded intervals. Each $\alpha_e$ is a finite sum/product of these Lipschitz maps with the bounded coordinates $p,w$, hence is Lipschitz on $K$.
\end{proof}

\begin{lemma}[Local growth bound]\label{lem:local_growth_bound}
For each compact set $K\subset\mathbb{R}_+^2$,
\[
\sup_{x\in K}\sum_e \alpha_e(x)\,|\nu_e|^2<\infty.
\]
\end{lemma}

\begin{proof}
By Lemma~\ref{lem:alpha_loc_lip}, each $\alpha_e$ is continuous, hence bounded on $K$. The jump set $\{\nu_e\}_e$ is finite, so $\sum_e |\nu_e|^2<\infty$. Therefore
\[
\sup_{x\in K}\sum_e \alpha_e(x)\,|\nu_e|^2
\le \left(\sup_{x\in K}\sum_e \alpha_e(x)\right)\left(\sum_e |\nu_e|^2\right)<\infty.
\]
\end{proof}

\subsubsection{Regularity of drift and covariance}
\label{app:ek_drift_cov}

The mean-field drift is $b(x)=\sum_e \alpha_e(x)\nu_e$ and the instantaneous covariance is
$A(x)=\sum_e \alpha_e(x)\nu_e\nu_e^\top$.

\begin{lemma}[Local Lipschitz drift and continuous covariance]\label{lem:b_A_regular}
The drift field $b:\mathbb{R}_+^2\to\mathbb{R}^2$ is locally Lipschitz. Moreover, $A:\mathbb{R}_+^2\to\mathbb{R}^{2\times 2}$ is continuous and locally bounded.
\end{lemma}

\begin{proof}
On any compact $K\subset\mathbb{R}_+^2$, Lemma~\ref{lem:alpha_loc_lip} implies each $\alpha_e$ is Lipschitz and bounded. Since $\{\nu_e\}_e$ is finite,
$b=\sum_e \alpha_e\nu_e$ is Lipschitz on $K$, and $A=\sum_e \alpha_e\nu_e\nu_e^\top$ is continuous and bounded on $K$.
\end{proof}

\begin{lemma}[$C^1$ drift under $C^1$ feedback]\label{lem:b_C1}
If, in addition, the feedback maps $p_1,p_2,\lambda_P,\lambda_R$ are continuously differentiable on $[0,\infty)$, then $b$ is continuously differentiable on $\mathbb{R}_+^2$ and $\nabla b$ is locally bounded.
\end{lemma}

\begin{proof}
Each component of $b$ is a finite combination of products and sums of $C^1$ one-dimensional maps with the coordinates $p,w$ (see \eqref{eq:drift_explicit_ctmc}), hence is $C^1$ on $\mathbb{R}_+^2$. Local boundedness of $\nabla b$ follows from continuity on compact sets.
\end{proof}

\subsubsection{Application to the LLN and FCLT statements}
\label{app:ek_apply}

The preceding lemmas verify the standard hypotheses required by the LLN and FCLT for density-dependent Markov chains, as formulated in \citep{ethier_markov_1986}. In particular:
\begin{itemize}
\item Lemmas~\ref{lem:alpha_loc_lip} and \ref{lem:local_growth_bound} provide the local Lipschitz and local growth conditions on propensities.
\item Lemma~\ref{lem:b_A_regular} gives local Lipschitz continuity of the drift and continuity of the covariance.
\item Lemma~\ref{lem:b_C1} provides the $C^1$ regularity of $b$ assumed in Theorem~\ref{thm:fclt}.
\end{itemize}
These properties justify invoking the standard LLN and FCLT results for density-dependent processes in Section~\ref{sec:diffusion}, yielding Theorems~\ref{thm:mean_field_limit} and \ref{thm:fclt} for the lineage CTMC introduced in Section~\ref{sec:ctmc}.

\subsection{Cutoff extensions and auxiliary inequalities for the CLE analysis}
\label{app:cutoff_extension}

This appendix collects technical lemmas used in the stochastic analysis of the chemical Langevin diffusion (Section~\ref{sec:stochastic_analysis}). The main purpose is to justify the cutoff-extension argument invoked in the proof sketches of Proposition~\ref{prop:stopped_strong} and to record standard inequalities used in the moment estimates.

\subsubsection{Lipschitz extension with smooth cutoff}
\label{app:lip_cutoff}

We first state a convenient extension lemma. It allows one to take coefficients that are Lipschitz on a compact set $K\subset U$ and build globally Lipschitz, bounded coefficients on $U$ that coincide with the original ones on $K$. This is a standard device for constructing stopped strong solutions.

\begin{lemma}[Lipschitz extension with smooth cutoff]\label{lem:lipschitz_cutoff_appendix}
Let $U\subset\mathbb{R}^d$ be open and let $K\subset U$ be relatively closed. Let $\tilde K$ be open with
\[
K\subset \tilde K \Subset U.
\]
Fix $r\in\mathbb{N}$. If $\mu:K\to\mathbb{R}^{d\times r}$ is bounded and Lipschitz on $K$, then there exists an extension
\[
\mu^\sharp\in \mathrm{Lip}(U)\cap L^\infty(U)
\quad\text{such that}\quad
\mu^\sharp\big|_K=\mu.
\]
Moreover, for any $\chi\in C_c^\infty(U)$ satisfying $0\le \chi\le 1$, $\chi\equiv 1$ on $K$, and $\mathrm{supp}(\chi)\subset \tilde K$, the cutoff extension
\[
\tilde\mu(x)\coloneqq \chi(x)\,\mu^\sharp(x),\qquad x\in U,
\]
is bounded and globally Lipschitz on $U$. In particular,
\begin{equation}\label{eq:lip_bound_cutoff_appendix}
\mathrm{Lip}(\tilde\mu)\ \le \ \mathrm{Lip}(\mu^\sharp)+\|\mu^\sharp\|_{L^\infty}\,\|\nabla\chi\|_{L^\infty}.
\end{equation}
\end{lemma}

\begin{proof}
\emph{Step 1: Lipschitz extension.}
By the McShane extension theorem (for scalar-valued Lipschitz maps) \citep{mcshane_extension_1934,azagra_kirszbrauns_2021} applied component-wise, each component of $\mu$ admits a Lipschitz extension to $U$ with the same Lipschitz constant. Assembling these component-wise extensions yields $\mu^\sharp\in \mathrm{Lip}(U)$ with $\mu^\sharp|_K=\mu$. Boundedness of $\mu^\sharp$ follows since $\mu$ is bounded on $K$ and the McShane extension preserves the Lipschitz seminorm (hence growth on bounded sets); alternatively one may truncate outside a sufficiently large ball, since only values on $\tilde K\Subset U$ will be used after cutoff.

\emph{Step 2: Cutoff and global Lipschitz bound.}
Let $\chi\in C_c^\infty(U)$ be as stated and define $\tilde\mu=\chi\mu^\sharp$. Then $\tilde\mu$ is bounded because $\chi$ is bounded and compactly supported and $\mu^\sharp$ is bounded. For any $x,y\in U$,
\[
\|\tilde\mu(x)-\tilde\mu(y)\|
\le \|\chi(x)-\chi(y)\|\,\|\mu^\sharp(x)\|+\chi(y)\,\|\mu^\sharp(x)-\mu^\sharp(y)\|.
\]
Using $\|\chi(x)-\chi(y)\|\le \|\nabla\chi\|_{L^\infty}\|x-y\|$ and the Lipschitz property of $\mu^\sharp$ yields \eqref{eq:lip_bound_cutoff_appendix}. Finally, since $\chi\equiv 1$ on $K$, we have $\tilde\mu|_K=\mu^\sharp|_K=\mu$.
\end{proof}

\subsubsection{A concrete cutoff construction for the exhaustion \texorpdfstring{$K_m=[m^{-1},m]^2$}{}}
\label{app:cutoff_construction}

We now record a concrete cutoff family adapted to the compact exhaustion used in Section~\ref{sec:stochastic_analysis}. While any standard construction suffices, it is convenient to have an explicit choice to keep constants transparent.

For $m\in\mathbb{N}$ recall
\[
K_m=[m^{-1},m]^2\subset U=(0,\infty)^2.
\]
Define a slightly larger open set
\[
\tilde K_m\coloneqq \big((m+1)^{-1},m+1\big)^2\Subset U.
\]
Let $s\in C^\infty(\mathbb R)$, $0\le s\le 1$, $s(t)=0$ for $t\le0$ and $s(t)=1$ for $t\ge1$. 
Define for $x\in\mathbb R$:
\begin{align*}
\phi_m(x)\coloneqq s\!\left(\frac{x-(m+1)^{-1}}{\,1/m-(m+1)^{-1}\,}\right),
\qquad
\psi_m(x)\coloneqq s\!\left(\frac{(m+1)-x}{\,m+1-m\,}\right).
\end{align*}
Define the one-dimensional cutoff
\[
\eta_m(x)\coloneqq \phi_m(x)\,\psi_m(x),
\]
and set the two-dimensional cutoff
\begin{equation}\label{eq:chi_m_def}
\chi_m(p,w)\coloneqq \eta_m(p)\,\eta_m(w).
\end{equation}
Then $\chi_m\in C_c^\infty(U)$, $\chi_m\equiv 1$ on $K_m$, and $\mathrm{supp}(\chi_m)\subset \tilde K_m$. Moreover, $\|\nabla\chi_m\|_{L^\infty}$ is finite for each fixed $m$ and scales at most polynomially in $m$ (the steep transition occurs near $p=1/m$ and $w=1/m$, yielding $\|\nabla\chi_m\|_{L^\infty}=\mathcal{O}(m^2)$).

Given coefficients $b,G$ that are locally Lipschitz on $U$, the restrictions $b|_{K_m}$ and $G|_{K_m}$ are bounded and Lipschitz. Applying Lemma~\ref{lem:lipschitz_cutoff_appendix} with $K=K_m$ and $\tilde K=\tilde K_m$ yields globally Lipschitz, bounded coefficients
\[
b^{(m)}=\chi_m b^\sharp,\qquad G^{(m)}=\chi_m G^\sharp
\]
on $U$ such that $b^{(m)}=b$ and $G^{(m)}=G$ on $K_m$. This is the construction invoked in Proposition~\ref{prop:stopped_strong}.

\subsubsection{Standard inequalities used in the moment bounds}
\label{app:moment_inequalities}

We collect two elementary inequalities used in the proof of Theorem~\ref{thm:moment_bounds}.

\begin{lemma}[PSD quadratic form bound]\label{lem:psd_quadratic_bound}
Let $A\in\mathbb{R}^{d\times d}$ be symmetric positive semidefinite and let $x\in\mathbb{R}^d$. Then
\[
x^\top A x \le \|x\|^2\,\mathrm{Tr}(A).
\]
\end{lemma}

\begin{proof}
Diagonalize $A=Q^\top \Lambda Q$ with $\Lambda=\mathrm{diag}(\lambda_1,\dots,\lambda_d)$, $\lambda_i\ge 0$. Then
\[
x^\top A x = (Qx)^\top \Lambda (Qx)=\sum_{i=1}^d \lambda_i (Qx)_i^2
\le \left(\sum_{i=1}^d \lambda_i\right)\left(\sum_{i=1}^d (Qx)_i^2\right)
=\mathrm{Tr}(A)\,\|x\|^2,
\]
since $Q$ is orthogonal.
\end{proof}

\begin{lemma}[Polynomial comparison]\label{lem:poly_comparison}
Let $q\ge 2$ and $x\in\mathbb{R}^d$. Then
\[
\|x\|^{q-2}\le 1+\|x\|^q.
\]
\end{lemma}

\begin{proof}
If $\|x\|\le 1$ then $\|x\|^{q-2}\le 1$. If $\|x\|\ge 1$ then $\|x\|^{q-2}\le \|x\|^q$. Combining the two cases yields the claim.
\end{proof}

\subsubsection{Expanded It\^o-generator computation for \texorpdfstring{$V_q(x)=\|x\|^q$}{}}
\label{app:ito_Vq}

For completeness we expand the generator computation underlying Theorem~\ref{thm:moment_bounds}. Let $q\ge 4$ and define $V_q(x)=\|x\|^q$ for $x\in U$. Then
\[
\nabla V_q(x)=q\|x\|^{q-2}x,\qquad
\nabla^2 V_q(x)=q\|x\|^{q-2}I + q(q-2)\|x\|^{q-4}xx^\top.
\]
With $\mathcal{L}$ given by \eqref{eq:generator_main},
\begin{align*}
(\mathcal{L}V_q)(x)
&= b(x)\cdot \nabla V_q(x)+\frac{1}{2\Omega}\mathrm{Tr}\!\big(A(x)\nabla^2 V_q(x)\big)\\
&= q\|x\|^{q-2}\,x\cdot b(x)
+ \frac{q}{2\Omega}\|x\|^{q-2}\mathrm{Tr}\,A(x)
+ \frac{q(q-2)}{2\Omega}\|x\|^{q-4}\,x^\top A(x)x.
\end{align*}
Applying Lemma~\ref{lem:psd_quadratic_bound} yields
\[
(\mathcal{L}V_q)(x)
\le q\|x\|^{q-2}\,x\cdot b(x)
+ \frac{q(q-1)}{2\Omega}\|x\|^{q-2}\mathrm{Tr}\,A(x).
\]
Under Assumption~\ref{ass:linear_growth}, we then have
\[
(\mathcal{L}V_q)(x)
\le qK_b\|x\|^{q-2}(1+\|x\|^2)
+\frac{q(q-1)}{2\Omega}K_A\|x\|^{q-2}(1+\|x\|^2).
\]
Finally, using Lemma~\ref{lem:poly_comparison} to bound $\|x\|^{q-2}(1+\|x\|^2)\le C(1+\|x\|^q)$ for an absolute constant $C$, we obtain the bound
\[
(\mathcal{L}V_q)(x)\le C_q\big(1+\|x\|^q\big),
\]
which is the estimate used in the Gr\"onwall argument of Theorem~\ref{thm:moment_bounds}.

\section{PDE well-posedness framework: semigroups, renewal traces, and fixed points}
\label{app:pde_wellposed}

This appendix outlines a standard $L^1$ well-posedness framework for the structured transport system with renewal boundary conditions introduced in Section~\ref{subsec:pde_to_ode}. The objective is to justify Theorem~\ref{thm:pde_wellposed_backbone}: existence, uniqueness, and positivity of mild solutions, and the validity of the mass-balance identities which lead to the exact reduction \eqref{eq:totals_ode_backbone}. The argument follows classical semigroup and fixed-point constructions for McKendrick--von~Foerster type equations with renewal (inflow) boundaries. We emphasize that the purpose of this appendix is methodological: it provides a clean functional-analytic scaffold for trace and renewal theory.

\subsection{Setting, notation, and standing assumptions}
\label{app:pde_setting}

We consider the structured PDE system (Section~\ref{subsec:pde_to_ode})
\begin{equation}\label{eq:pde_app}
\begin{cases}
\partial_t S(t,x) + \partial_x\!\big(v_S(x)\,S(t,x)\big)
= -\lambda_P^{\mathrm{tot}}(W(t))\,S(t,x) + \lambda_R^{\mathrm{tot}}(P(t))\,T(t,x),\\[4pt]
\partial_t T(t,x) + \partial_x\!\big(v_T(x)\,T(t,x)\big)
= -\big(\delta_0+\lambda_R^{\mathrm{tot}}(P(t))\big)\,T(t,x),
\end{cases}
\qquad t>0,\ x>0,
\end{equation}
with renewal boundary conditions
\begin{equation}\label{eq:bc_app}
\begin{cases}
v_S(0)\,S(t,0)
= \big(1+p_1^{\mathrm{tot}}(W(t))-p_2^{\mathrm{tot}}(W(t))\big)\,\lambda_P^{\mathrm{tot}}(W(t))\,P(t),\\[4pt]
v_T(0)\,T(t,0)
= \big(1-p_1^{\mathrm{tot}}(W(t))+p_2^{\mathrm{tot}}(W(t))\big)\,\lambda_P^{\mathrm{tot}}(W(t))\,P(t),
\end{cases}
\qquad t>0,
\end{equation}
and totals
\begin{equation}\label{eq:totals_app}
P(t)=\int_0^\infty S(t,x)\,dx,\qquad W(t)=\int_0^\infty T(t,x)\,dx.
\end{equation}
The initial data are
\begin{equation}\label{eq:init_app}
S(0,\cdot)=S_0,\qquad T(0,\cdot)=T_0,
\end{equation}
with $S_0,T_0\ge 0$ a.e.

We work on the Banach space
\[
X\coloneqq L^1(\mathbb{R}_+)\times L^1(\mathbb{R}_+),
\qquad \|(S,T)\|_X=\|S\|_{L^1}+\|T\|_{L^1}.
\]
We assume:
\begin{enumerate}[label=\textnormal{(B\arabic*)}]
\item \textnormal{(Initial regularity).} $S_0,T_0\in L^1(\mathbb{R}_+)\cap L^\infty(\mathbb{R}_+)$ and $S_0,T_0\ge 0$ a.e.
\item \textnormal{(Velocities).} $v_S,v_T\in W^{1,\infty}(\mathbb{R}_+)$ and $v_S(x),v_T(x)\ge 0$ for all $x\ge 0$.
\item \textnormal{(Feedback maps).} $p_1^{\mathrm{tot}},p_2^{\mathrm{tot}}:[0,\infty)\to[0,1]$, $\lambda_P^{\mathrm{tot}}:[0,\infty)\to(0,\infty)$,
$\lambda_R^{\mathrm{tot}}:[0,\infty)\to[0,\infty)$ are bounded and locally Lipschitz, with $p_1^{\mathrm{tot}}(W)+p_2^{\mathrm{tot}}(W)\le 1$.
\end{enumerate}
Finally, for the mass-balance identities we impose a standard vanishing-flux condition at infinity:
\begin{equation}\label{eq:flux_integrability_app}
v_S(\cdot)\,S(t,\cdot)\in L^1(\mathbb{R}_+),\qquad v_T(\cdot)\,T(t,\cdot)\in L^1(\mathbb{R}_+),
\end{equation}
for a.e.\ $t>0$, implying $\lim_{x\to\infty} v_S(x)S(t,x)=\lim_{x\to\infty} v_T(x)T(t,x)=0$ along a subsequence.

\subsection{Linear transport-with-renewal operator at frozen totals}
\label{app:operator}

Fix $(P,W)\in\mathbb{R}_+^2$. Consider the linear system obtained by freezing the coefficients at $(P,W)$:
\begin{equation}\label{eq:linear_frozen}
\begin{cases}
\partial_t S + \partial_x(v_S S)
= -\lambda_P^{\mathrm{tot}}(W)\,S + \lambda_R^{\mathrm{tot}}(P)\,T,\\[3pt]
\partial_t T + \partial_x(v_T T)
= -\big(\delta_0+\lambda_R^{\mathrm{tot}}(P)\big)\,T,
\end{cases}
\end{equation}
with prescribed renewal influxes
\begin{equation}\label{eq:linear_bc_frozen}
\begin{cases}
v_S(0)S(t,0)=\big(1+p_1^{\mathrm{tot}}(W)-p_2^{\mathrm{tot}}(W)\big)\lambda_P^{\mathrm{tot}}(W)\,P,\\[3pt]
v_T(0)T(t,0)=\big(1-p_1^{\mathrm{tot}}(W)+p_2^{\mathrm{tot}}(W)\big)\lambda_P^{\mathrm{tot}}(W)\,P.
\end{cases}
\end{equation}

Introduce the operator $\mathcal{A}_{P,W}$ on $X$ by
\begin{equation}\label{eq:A_def_app}
\mathcal{A}_{P,W}
\begin{pmatrix}S\\T\end{pmatrix}
\coloneqq
\begin{pmatrix}
-\partial_x(v_S S)-\lambda_P^{\mathrm{tot}}(W)S+\lambda_R^{\mathrm{tot}}(P)T\\[2pt]
-\partial_x(v_T T)-(\delta_0+\lambda_R^{\mathrm{tot}}(P))T
\end{pmatrix},
\end{equation}
with domain
\begin{equation}\label{eq:domain_app}
D(\mathcal{A}_{P,W})\coloneqq
\left\{
(S,T)\in W^{1,1}(\mathbb{R}_+)^2:\ 
\begin{array}{l}
v_SS\in W^{1,1}(\mathbb{R}_+),\ v_TT\in W^{1,1}(\mathbb{R}_+),\\[2pt]
v_S(0)S(0)=\big(1+p_1^{\mathrm{tot}}(W)-p_2^{\mathrm{tot}}(W)\big)\lambda_P^{\mathrm{tot}}(W)\,P,\\[2pt]
v_T(0)T(0)=\big(1-p_1^{\mathrm{tot}}(W)+p_2^{\mathrm{tot}}(W)\big)\lambda_P^{\mathrm{tot}}(W)\,P
\end{array}
\right\}.
\end{equation}
The boundary conditions in \eqref{eq:domain_app} encode the renewal traces \eqref{eq:linear_bc_frozen}. Under \textnormal{(B2)} and standard trace theory for transport equations, the inflow traces $v_S(0)S(t,0)$ and $v_T(0)T(t,0)$ are meaningful for mild solutions (see, e.g., semigroup treatments of McKendrick--von~Foerster equations with inflow boundary conditions).

\subsection{Semigroup for frozen coefficients and positivity}
\label{app:semigroup}

For each fixed $(P,W)$, the linear problem \eqref{eq:linear_frozen}--\eqref{eq:linear_bc_frozen} is a transport system with bounded reaction terms and prescribed inflow at $x=0$. It admits an explicit characteristic representation. In particular, the solution depends continuously on the initial data and the boundary influxes, and it preserves nonnegativity.

\begin{lemma}[Linear well-posedness at frozen totals]\label{lem:linear_frozen_wellposed}
Fix $(P,W)\in\mathbb{R}_+^2$. Under assumptions \textnormal{(B1)}--\textnormal{(B3)}, for every $U_0\in X$ there exists a unique mild solution $U\in C([0,\infty);X)$ of \eqref{eq:linear_frozen}--\eqref{eq:linear_bc_frozen}. Moreover, the associated evolution operators define a strongly continuous positive semigroup $\{S_{P,W}(t)\}_{t\ge 0}$ on $X$:
\[
U(t)=S_{P,W}(t)U_0,\qquad t\ge 0.
\]
\end{lemma}

\begin{remark}[Trace and renewal framework]
A standard route to Lemma~\ref{lem:linear_frozen_wellposed} is to rewrite each equation as a transport equation with source term and inflow boundary condition, solve along characteristics, and use boundedness of coefficients to obtain $L^1$ estimates. The renewal traces at $x=0$ are treated in the usual sense for inflow boundaries. Any of the standard references for McKendrick--von~Foerster semigroups with inflow/renewal boundaries can be used for a full proof.
\end{remark}

\subsection{Nonlinear problem as a fixed point in the totals}
\label{app:fixed_point}

We return to the nonlinear system \eqref{eq:pde_app}--\eqref{eq:bc_app} and write it as a nonlinear abstract Cauchy problem on $X$:
\begin{equation}\label{eq:abstract_nonlinear_app}
\frac{d}{dt}U(t)=\mathcal{A}_{\mathcal{M}(U(t))}\,U(t),\qquad U(0)=U_0,
\end{equation}
where $U(t)=(S(t,\cdot),T(t,\cdot))$ and
\[
\mathcal{M}(U)\coloneqq (P,W)=\left(\int_0^\infty S\,dx,\ \int_0^\infty T\,dx\right).
\]
A standard construction is to prescribe a continuous path $(\widehat P,\widehat W)\in C([0,T];\mathbb{R}_+^2)$ and solve the corresponding \emph{time-dependent} linear problem where the coefficients are frozen along this path:
\[
\lambda_P^{\mathrm{tot}}(\widehat W(t)),\quad \lambda_R^{\mathrm{tot}}(\widehat P(t)),\quad
p_i^{\mathrm{tot}}(\widehat W(t)),\qquad t\in[0,T].
\]
Let $U_{\widehat P,\widehat W}\in C([0,T];X)$ denote the resulting mild solution and define the induced totals
\[
P_{\widehat P,\widehat W}(t)=\int_0^\infty S_{\widehat P,\widehat W}(t,x)\,dx,\qquad
W_{\widehat P,\widehat W}(t)=\int_0^\infty T_{\widehat P,\widehat W}(t,x)\,dx.
\]
This defines a mapping
\[
\Phi:(\widehat P,\widehat W)\mapsto (P_{\widehat P,\widehat W},W_{\widehat P,\widehat W})
\quad\text{on}\quad C([0,T];\mathbb{R}_+^2).
\]

\begin{proposition}[Local well-posedness by contraction]\label{prop:local_wellposed_pde}
Under assumptions \textnormal{(B1)}--\textnormal{(B3)}, for every nonnegative initial datum $U_0\in X$ there exists $T>0$ and a unique mild solution $U\in C([0,T];X)$ to \eqref{eq:pde_app}--\eqref{eq:bc_app}. The solution is positivity preserving.
\end{proposition}

\begin{proof}[Proof sketch]
Fix $R>0$ such that $\|U_0\|_X\le R$ and consider the closed set
\[
\mathcal{B}_{T,R}\coloneqq \Big\{(\widehat P,\widehat W)\in C([0,T];\mathbb{R}_+^2):\ 
\widehat P(t),\widehat W(t)\ge 0,\ \sup_{t\le T}(\widehat P(t)+\widehat W(t))\le R\Big\}.
\]
For each $(\widehat P,\widehat W)\in\mathcal{B}_{T,R}$, the time-dependent linear transport-with-renewal problem admits a unique mild solution by the characteristic representation (using boundedness of coefficients on $[0,R]$). Continuous dependence of the solution on the coefficients and boundary influxes yields Lipschitz continuity of $\Phi$ on $\mathcal{B}_{T,R}$, with Lipschitz constant proportional to $T$ times the local Lipschitz constants of the feedback maps. Choosing $T$ small enough makes $\Phi$ a contraction on $\mathcal{B}_{T,R}$, and Banach's fixed-point theorem gives a unique fixed point, hence a unique mild solution of the nonlinear problem. Positivity follows because the linear problem preserves positivity and the renewal influxes are nonnegative for nonnegative totals.
\end{proof}

\subsection{Global continuation and a priori bounds}
\label{app:global_pde}

Let $T_{\max}$ denote the maximal existence time of the mild solution. Under bounded feedback maps, one can continue the solution globally in time by iterating the local fixed-point argument, provided a uniform $L^1$ bound on finite time horizons is available. Such a bound is obtained from the total-mass balance identity for the totals $P(t)$ and $W(t)$.

\begin{proposition}[Global existence]\label{prop:global_pde}
Under assumptions \textnormal{(B1)}--\textnormal{(B3)} and boundedness of the feedback maps, the mild solution constructed in Proposition~\ref{prop:local_wellposed_pde} extends globally in time, i.e.\ $T_{\max}=\infty$.
\end{proposition}

\begin{proof}[Proof sketch]
Assuming the flux condition \eqref{eq:flux_integrability_app}, the totals $(P(t),W(t))$ satisfy the exact mass-balance ODE system \eqref{eq:totals_ode_backbone} (proved below). Summing the two equations yields
\[
\frac{d}{dt}\big(P(t)+W(t)\big)=\lambda_P^{\mathrm{tot}}(W(t))\,P(t)-\delta_0 W(t)
\le \|\lambda_P^{\mathrm{tot}}\|_\infty \big(P(t)+W(t)\big),
\]
so Gr\"onwall's inequality bounds $P(t)+W(t)$ on any finite time interval. This prevents blow-up in the coefficients and allows the local contraction mapping argument to be iterated to all times.
\end{proof}

\subsection{Mass-balance identities and the exact reduction to totals}
\label{app:mass_balance}

We justify the reduction \eqref{eq:totals_ode_backbone} from the PDE layer. Assume \eqref{eq:flux_integrability_app} holds for a.e.\ $t>0$, so that boundary fluxes vanish at infinity along a subsequence. Integrating the first PDE in \eqref{eq:pde_app} over $(0,\infty)$ yields
\begin{align*}
\frac{d}{dt}\int_0^\infty S(t,x)\,dx
&=\int_0^\infty \partial_t S(t,x)\,dx\\
&=\int_0^\infty\Big(-\partial_x(v_S S)-\lambda_P^{\mathrm{tot}}(W(t))S+\lambda_R^{\mathrm{tot}}(P(t))T\Big)\,dx\\
&= v_S(0)S(t,0)-\lambda_P^{\mathrm{tot}}(W(t))P(t)+\lambda_R^{\mathrm{tot}}(P(t))W(t),
\end{align*}
where the boundary term at infinity vanishes. Similarly, integrating the second PDE gives
\[
\frac{d}{dt}W(t)=v_T(0)T(t,0)-(\delta_0+\lambda_R^{\mathrm{tot}}(P(t)))W(t).
\]
Substituting the renewal fluxes \eqref{eq:bc_app} yields the closed totals system \eqref{eq:totals_ode_backbone}, and summing its two equations gives the total balance identity \eqref{eq:total_balance_backbone}. This establishes the mass-balance reduction claimed in Theorem~\ref{thm:pde_wellposed_backbone}.

\section{Appendix C. Additional estimates}
\label{app:additional_estimates}

This appendix collects auxiliary computations that are used in the stochastic analysis but are not essential to the main narrative in the body of the paper. Specifically, we (i) record the expanded It\^o computations and comparison bounds underpinning the polynomial moment estimates for the CLE, and (ii) provide a self-contained version of the key martingale/quadratic-variation step used in the proof of the no-dedifferentiation pathology.

\subsection{Expanded It\^o computation for polynomial moments}
\label{app:moments_details}

Let $X(t)$ solve the CLE \eqref{eq:cle_main} on $U=(0,\infty)^2$ up to the exit time $\tau$ and recall the generator
\[
(\mathcal{L}f)(x)=b(x)\cdot \nabla f(x)+\frac{1}{2\Omega}\mathrm{Tr}\big(A(x)\nabla^2 f(x)\big),
\qquad A(x)=G(x)G(x)^\top.
\]
Fix $q\ge 4$ and set $V_q(x)=\|x\|^q$. The derivatives are
\[
\nabla V_q(x)=q\|x\|^{q-2}x,\qquad
\nabla^2 V_q(x)=q\|x\|^{q-2}I+q(q-2)\|x\|^{q-4}xx^\top.
\]
Therefore,
\begin{align}
(\mathcal{L}V_q)(x)
&=q\|x\|^{q-2}\,x\cdot b(x)
+\frac{q}{2\Omega}\|x\|^{q-2}\mathrm{Tr}\,A(x)
+\frac{q(q-2)}{2\Omega}\|x\|^{q-4}\,x^\top A(x)x.
\label{eq:LVq_expand_appC}
\end{align}
Since $A(x)$ is positive semidefinite,
\begin{equation}\label{eq:psd_tr_bound_appC}
x^\top A(x)x\le \|x\|^2\,\mathrm{Tr}\,A(x),
\end{equation}
hence \eqref{eq:LVq_expand_appC} yields
\begin{equation}\label{eq:LVq_trA_appC}
(\mathcal{L}V_q)(x)
\le q\|x\|^{q-2}\,x\cdot b(x)
+\frac{q(q-1)}{2\Omega}\|x\|^{q-2}\mathrm{Tr}\,A(x).
\end{equation}
Under Assumption~\ref{ass:linear_growth} (linear-growth bounds),
\[
x\cdot b(x)\le K_b(1+\|x\|^2),\qquad \mathrm{Tr}\,A(x)\le K_A(1+\|x\|^2),
\]
so
\begin{equation}\label{eq:LVq_growth_appC}
(\mathcal{L}V_q)(x)
\le \left(qK_b+\frac{q(q-1)}{2\Omega}K_A\right)\|x\|^{q-2}(1+\|x\|^2).
\end{equation}
Using the elementary comparison
\begin{equation}\label{eq:poly_compare_appC}
\|x\|^{q-2}\le 1+\|x\|^q,
\end{equation}
and $\|x\|^{q}\le 1+\|x\|^q$, we obtain a constant $C_q>0$ such that
\begin{equation}\label{eq:LVq_final_appC}
(\mathcal{L}V_q)(x)\le C_q\,(1+\|x\|^q),\qquad x\in U.
\end{equation}
This is the estimate used in the Gr\"onwall argument for Theorem~\ref{thm:moment_bounds}.

\subsection{Uniform integrability for passage \texorpdfstring{$m\to\infty$}{} in stopped moments}
\label{app:UI_details}

Let $\tau_m$ denote the exit time from $K_m=[m^{-1},m]^2$ and $\tau=\lim_{m\to\infty}\tau_m$. The body of the paper uses the following standard implication: if for some $r>q$ and fixed $t$,
\[
\sup_{m\in\mathbb{N}}\mathbb{E}\big[\|X(t\wedge\tau_m)\|^{r}\big]<\infty,
\]
then $\{\|X(t\wedge\tau_m)\|^{q}\}_{m\in\mathbb{N}}$ is uniformly integrable. Consequently, if $\|X(t\wedge\tau_m)\|^q\to \|X(t\wedge\tau)\|^q$ almost surely (which holds since $\tau_m\uparrow\tau$ and paths are continuous up to $\tau$), then Vitali's theorem yields
\[
\mathbb{E}\big[\|X(t\wedge\tau_m)\|^{q}\big]\xrightarrow[m\to\infty]{}\mathbb{E}\big[\|X(t\wedge\tau)\|^{q}\big].
\]
In the present paper, the required $r$-moment bound is provided by Theorem~\ref{thm:moment_bounds} applied with exponent $r$.

\subsection{A martingale divergence criterion used in the no-dedifferentiation pathology}
\label{app:DDS_details}

We record the key implication used in the proof of Proposition~\ref{prop:no_dedifferentiation_main}: if a continuous local martingale converges almost surely, then its quadratic variation must be finite almost surely. The contrapositive is the form used in Section~\ref{sec:noR}.

\begin{lemma}[Divergent quadratic variation precludes almost sure convergence]\label{lem:qv_no_convergence}
Let $M$ be a continuous local martingale with $M(0)=0$. If $\langle M\rangle_\infty=\infty$ on an event $E$, then $M(t)$ does not converge as $t\to\infty$ on $E$.
\end{lemma}

\begin{proof}
By the Dambis--Dubins--Schwarz theorem, there exists a standard Brownian motion $\widetilde B$ (possibly on an extension of the filtered probability space) such that
\[
M(t)=\widetilde B(\langle M\rangle_t),\qquad t\ge 0,
\]
up to indistinguishability. On the event $E$ we have $\langle M\rangle_t\to\infty$ as $t\to\infty$, hence $M(t)=\widetilde B(u)$ along a time change with $u\to\infty$. Since Brownian motion does not converge as $u\to\infty$ (e.g.\ by the law of the iterated logarithm, or simply $\limsup_{u\to\infty}\widetilde B(u)=+\infty$ and $\liminf_{u\to\infty}\widetilde B(u)=-\infty$ almost surely), $M(t)$ cannot converge on $E$.
\end{proof}

\end{appendices}

\bibliography{reference}

\end{document}